\documentclass[journal]{IEEEtran}
\usepackage{amsmath,graphicx,url,times}
\usepackage{color}
\usepackage{lipsum}
\usepackage[numbers,sort&compress]{natbib}
\usepackage{subcaption}

\DeclareUnicodeCharacter{0306}{}

\usepackage{enumitem}
\usepackage{booktabs,tabularx}
\newcolumntype{C}{>{\centering\arraybackslash}X} % centered version of "X" type
\usepackage{todonotes}

\usepackage{NotationMacros}
\usepackage{AllpassFDNMacros}
\graphicspath{{./graphics/}}

\title{Allpass Feedback Delay Networks}

\author{Sebastian J. Schlecht,~\IEEEmembership{Senior Member,~IEEE}
        }% <-this % stops a space
        
\begin{document} 

\maketitle

\begin{sloppy}

\begin{abstract}
In the 1960s, Schroeder and Logan introduced delay line-based allpass filters, which are still popular due to their computational efficiency and versatile applicability in artificial reverberation, decorrelation, and dispersive system design. In this work, we extend the theory of allpass systems to any arbitrary connection of delay lines, namely feedback delay networks (FDNs). We present a characterization of uniallpass FDNs, i.e., FDNs, which are allpass for an arbitrary choice of delays. Further, we develop a solution to the completion problem, i.e., given an FDN feedback matrix to determine the remaining gain parameters such that the FDN is allpass. Particularly useful for the completion problem are feedback matrices, which yield a homogeneous decay of all system modes. Finally, we apply the uniallpass characterization to previous FDN designs, namely, Schroeder's series allpass and Gardner's nested allpass for single-input, single-output systems, and, Poletti's unitary reverberator for multi-input, multi-output systems and demonstrate the significant extension of the design space.
\end{abstract}


\section{Introduction}
\label{sec:intro}

%\subsection{Related Structures}

\IEEEPARstart{A}{}
llpass filters preserve the signal's energy and only alter the signal phase \cite{Regalia:1988cj}. \schlecsn{Schroeder and Logan generalized the first-order allpass filter replacing the single delay with a delay line \cite{Schroeder:1961ke}.
%delay-based allpass filters in the 1960s to create ``colorless'' artificial reverberation. 
A decade later, Gerzon generalized delay line-based filters, e.g., feedback comb filters, to feedback delay networks (FDNs) \cite{Gerzon:1971tu} and the single-input, single-output (SISO) allpass structure to multi-input, multi-output (MIMO) allpass networks \cite{Gerzon:1976fm}.}

%and they remain popular due to their computational efficiency and wide applicability in acoustics, audio signal processing, and dispersive system design \cite{Valimaki:2012jv}

%An FDN essentially consists of a set of delay lines interconnected via a feedback matrix \cite{Gerzon:1971tu}. 
\schlecsn{FDNs generalize the well-known state space representation by replacing single time steps with different vector time steps, see Fig.~\ref{fig:mimoFDN}.} 
%FDNs can have single or multiple input and output channels distributed by the input, output, and direct gains. 
FDNs have well-established system properties such as losslessness and stability \cite{Rocchesso:1997fv,Schlecht:2017jt}, decay control \cite{Jot:1991tq, Prawda:2019tq}, impulse response density \cite{DeSena:2015bb,Schlecht:2017il}, and, modal distribution \cite{Schlecht:2019uj}. 
%Compared to general high-order allpass filters \cite{Abel:2006ux}, FDNs are less flexible, but more computationally efficient. 
SISO allpass FDNs can be \schlecsn{composed} from simple allpass filters in series \cite{Schroeder:1961ke, Schlecht:2019sa} or by nesting \cite{Gardner:1992ks}. 
%to create more complex structures while retaining the allpass characteristic. 
Rocchesso and Smith also suggested an almost allpass FDN with equal delays in \cite[Th.~2]{Rocchesso:1997fv}. MIMO allpass filters can be similarly generated from simple unitary building blocks \cite{Gerzon:1976fm, Vaidyanathan:1989} or by generalizing the allpass lattice structure \cite{Poletti:1995tq}.

Both SISO and MIMO allpass FDNs were applied to a wide range of roles including: 1) increasing the echo density as preprocessing to an artificial reverberator \cite{Schroeder:1961ke,Valimaki:2012jv}; 2) increasing echo density of in the feedback loop of reverberators \cite{Vaananen:1997wj,Lokki:2001ly,Schlecht:2015hi,Werner.2020.Energy}; 3) decorrelation for widening the auditory image of a sound source \cite{Kendall:1995be,Abel:2019aa,Gribben:2020}; 4) as reverberator in electro-acoustic reverberation enhancement systems \cite{Lokki:2001ly,Poletti:1995tq,Poletti:2004hh,Schlecht:2016ta}; 5) linear dynamic range reduction \cite{Parker:2013fy, Belloch:2014} ; and 6) dispersive system design \cite{Abel:2006ux, Valimaki:2010wc, Parker:2011fn}. In the broader context of control theory, allpass FDNs are strongly related to Schur diagonal stability \cite{Kaszkurewicz:2000}, e.g., stability properties of asynchronous networks. \schlecsn{The characterization of allpass matrix-valued rational functions is closely connected to spectral factorization \cite{Baggio:2015.Factorization, Baggio:2017.SpectralDensity} and the notion of balanced realization of state space filters \cite{Regalia:1987fo,Mullis:1976.Roundoff,Hanzon:2000}.}

%In artificial reverberators and  decorrelators, the allpass characteristic is most important when used as short diffusing filters. When used with longer decay times, only the local comb-like spectrum is perceived. For reverberation enhancement systems and feedback delay networks, the allpass characteristic is mathematically important to maintain the system stability. 
%\todo{add unistable FDNs with diagonal stable matrices}

%%%%%%%%%%%%%%%%%%%%%%%
\begin{figure}[!tb]
\centering
\includegraphics[width=0.5\textwidth]{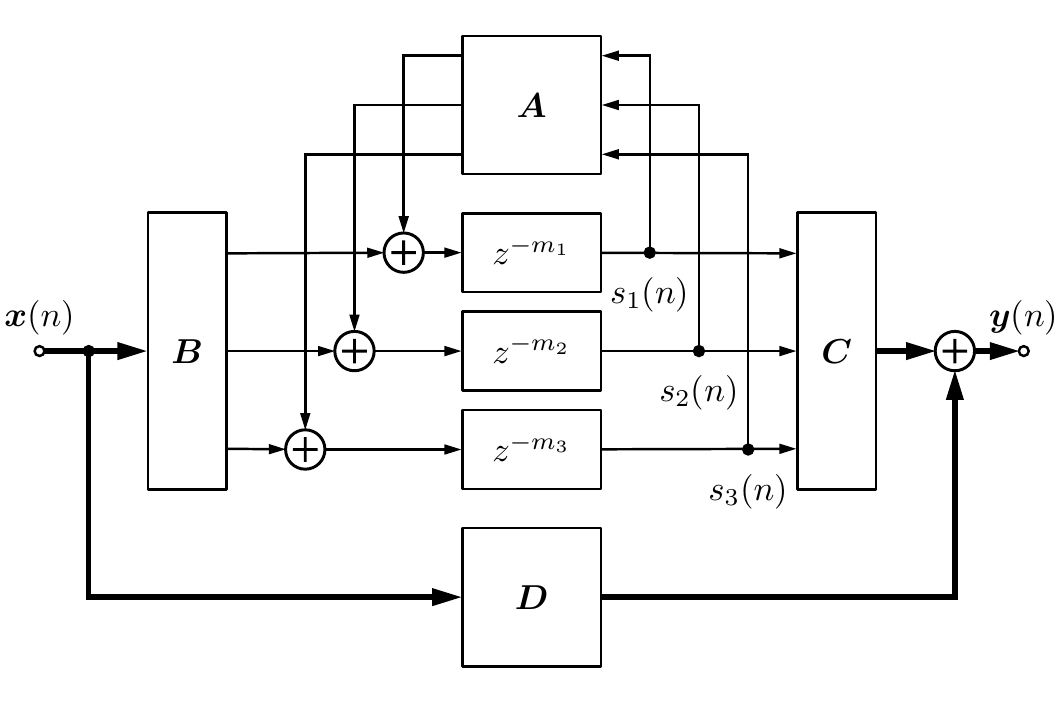}
\caption{MIMO feedback delay network (FDN) with three delay lines, i.e., $\matSize = 3$ and feedback matrix $\Fbm$. Thick lines indicate multiple channels, while thin lines indicate individual channels.}
\label{fig:mimoFDN}
\end{figure}
%%%%%%%%%%%%%%%%%%%%%%%

In this work, we extend the theory of allpass FDNs for both SISO and MIMO. In particular, we study \emph{uniallpass}\footnote{The term \emph{uniallpass} is introduced here with similar motivation as \emph{unilossless} feedback matrices in \cite{Schlecht:2017jt} which yields lossless FDNs regardless of delay lengths.} FDNs, i.e., FDNs, which are allpass for arbitrary delay lengths. While not all allpass FDNs are uniallpass \schlecsn{, e.g., see example in Section~\ref{sec:allpassNotUniallpass},} the more straightforward design criterion significantly extends practical filter structures.

The feedback matrix determines many filter properties of the FDN. Thus, it is often desirable to first design the feedback matrix and subsequently choose the input, output, and direct gains such that the resulting FDN is allpass. We refer to this procedure as the \emph{completion} problem. We call feedback matrices which have a solution to the completion problem as being \emph{allpass admissible}. A particularly useful class of feedback matrices are lossless mixing matrices in conjunction with diagonal delay-proportional absorption matrices. They result in \emph{homogeneous decay} of the impulse response, i.e., all system eigenvalues have the same magnitude \cite{Jot:1991tq}. The main contributions of this work are
\begin{itemize}
	\item \schlecsn{Improved sufficient condition for an FDN to be stable (Theorem~\ref{th:stableFDN}) in Section~\ref{sec:allpassFDN}}
	\item Sufficient and necessary conditions for SISO and MIMO FDNs to be uniallpass (Theorems~\ref{th:allpassFDN} and \ref{th:PM} in Section~\ref{sec:allpassFDN})
	\item Characterization of admissible feedback matrices in uniallpass FDN (Section~\ref{sec:admissible})
	\item Completion algorithms for uniallpass SISO and MIMO FDNs (Section~\ref{sec:generalCompletion})
	\item Characterization of uniallpass FDNs with homogeneous decay (Section~\ref{sec:homogeneous})
	\item Embedding of previous designs in the proposed characterization (Section~\ref{sec:application}). 
\end{itemize}
This work extends the design space of delay line-based allpass filters from a handful of known structures to a freely parametrizable extensive class. In particular, the solution of the completion problem allows to combine feedback matrix design with the allpass property and potentially improves application designs mentioned above. \schlecsn{A MATLAB implementation of all plots, examples and the completion algorithm are included in the FDN toolbox \cite{Schlecht.2020.FDNTB}\footnote{\protect{\url{https://github.com/SebastianJiroSchlecht/fdnToolbox}}}.}

%We further show that previous allpass FDN designs such as Schroeder's series allpass \cite{Schroeder:1960js}, Gardner's nested allpasses \cite{Gardner:1992ks}, and Poletti's unitary reverberator \cite{Poletti:1995tq} are uniallpass. At the same time, we also demonstrate that previous designs constitute only a fraction of all possible uniallpass FDNs.
 
%The remaining manuscript is structured as follows. Section~\ref{sec:priorArt} introduces FDN and allpass prior art and reviews a classic theorem on allpass state space systems. 

%In Section~\ref{sec:allpassFDN}, we characterize uniallpass FDNs. Section~\ref{sec:completion} presents a characterization of admissible feedback matrices and presents a completion algorithm. Section~\ref{sec:homogeneous} derives a solution for uniallpass FDNs with homogeneous decay. In Section~\ref{sec:application}, we give examples of the proposed method and comparison to previous allpass FDN designs.

%\todo{check all figures}

%\schlecsn{Could you briefly mention somewhere the potential advantages of the new class of allpass and uniallpass FDN systems?}

%\footnote{\schlecsn{For reproducibility, the MATLAB code to reproduce all data and figures are published at \protect{\url{https://github.com/SebastianJiroSchlecht/FrequencyDependentSchroederAllpass}}}.

\section{Problem Statement and Prior Art}
\label{sec:priorArt}

%In the following, we state the problem formulation of this work and review the prior art.
This section introduces FDN and allpass prior art and reviews a classic theorem on allpass state space systems.

\subsection{MIMO Feedback Delay Network}
\label{sec:MIMOFDN}

The MIMO FDN is given in the discrete-time domain by the difference equation in delay state space form \cite{Rocchesso:1997fv}, see Fig.~\ref{fig:mimoFDN}, 
\begin{equation}
\begin{aligned}
	\Outsignal(n) &= \Outgains \Ssstate(n) + \Directgains \, \Insignal(n), \\
	\Ssstate(n+\Delay) &= \Fbm \, \Ssstate(n) + \Ingains \, \Insignal(n),
\label{eq:FDNTimeDomain}
\end{aligned}
\end{equation}
where $\Insignal(n)$ and $\Outsignal(n)$ are the $\numInput \times 1$ input and $\numOutput \times 1$ output vectors at time sample $n$, respectively. The FDN dimension $\matSize$ is the number of delay lines. The FDN consists of the $\matSize \times \matSize$ \emph{feedback matrix} $\Fbm$, the $\matSize \times \numInput$ \emph{input gain} matrix $\Ingains$, the $\numOutput \times \matSize$  \emph{output gain} matrix $\Outgains$ and the $\numOutput \times \numInput$ \emph{direct gain} matrix $\Directgains$. The lengths of the $\matSize$ delay lines in samples are given by the vector $\Delay = [\delay_1, \dots, \delay_\matSize]$. The $\matSize \times 1$ vector $\Ssstate(n)$ denotes the delay-line outputs at time $n$. The vector argument notation $\Ssstate(n + \Delay)$ abbreviates the vector $[\ssstate_1(n+\delay_1), \dots, \ssstate_\matSize(n+\delay_\matSize)]$. \schlecsn{We focus with our results on FDNs with equal input and output channels, i.e., $\numIO = \numInput = \numOutput$ and real-valued filter coefficients.} We refer to an FDN where the number of delay lines is equal to the input and output channels as \emph{full MIMO}, i.e., $\numIO = \matSize$. A SISO FDN has $\numIO = 1$, which is emphasized by using vectors and scalars $\Ingain$, $\Outgain$ and $\directgain$ instead of matrices $\Ingains$, $\Outgains$ and $\Directgains$.

The $\numIO \times \numIO$ transfer function matrix of an FDN in the z-domain \cite{Rocchesso:1997fv} corresponding to \eqref{eq:FDNTimeDomain} is 
\begin{equation}
	\TF(z) = \Outgains \invp{\stdDelayArg{\inv{z}} - \Fbm} \Ingains + \Directgains,
	\label{eq:FDNtrans}
\end{equation}
where $\stdDelay =\diag{[z^{-\delay_1}, z^{-\delay_2},\dots,z^{-\delay_\matSize} ]}$ is the diagonal $\matSize \times \matSize$ delay matrix \cite{Jot:1991tq}. The system order is given by the sum of all delay units, i.e., $\N = \sum_{i=1}^\matSize \delay_i$ \cite{Rocchesso:1997fv}. For commonly used delays $\Delay$, the system order is much larger than the FDN size, i.e., $\N \gg \matSize$.

The transfer function matrix \eqref{eq:FDNtrans} can be stated as a rational polynomial \cite{Rocchesso:1997fv,Schlecht:2019uj}, i.e.,
\begin{equation}
	\TF(z) = \frac{\Gcqs(z)}{\gcp_{\Delay, \Fbm}(z)}, 
	\label{eq:FDNtransRational}
\end{equation}
where the denominator is a scalar-valued polynomial
\begin{equation}		
	\gcp_{\Delay, \Fbm}(z) = \detp{ \MatPoly(z) },
	\label{eq:gcp}
\end{equation}	
where $\det$ denotes the determinant and the loop transfer function is
\begin{equation}
	\MatPoly(z) = \stdDelayArg{\inv{z}} -  \Fbm.
	\label{eq:loopTransfer}
\end{equation}
The numerator is a matrix-valued expression with
\begin{equation}
\begin{aligned}
	\Gcqs(z) =
	\Directgains \detp{ \MatPoly(z) } + \Outgains \, \adj(\MatPoly(z)) \, \Ingains
	\label{eq:gcq},
\end{aligned}
\end{equation}
where $\adj( \Fbm )$ denotes the adjugate of $\Fbm$ \cite{Schlecht:2019uj}. 
%For brevity, we occasionally omit the parameters and write $\Gcq(z)$ and $\gcp(z)$. 
%If $\Fbm$ in an invertible matrix, then $ \adj(\Fbm) = \detp{\Fbm} \inv{\Fbm}$ \cite{Golub:1996wp}. 
The FDN system poles $\pole_\poleIndex$, where $1 \leq i \leq \N$, are the roots of the generalized characteristic polynomial (GCP) $\gcp_{\Delay, \Fbm}(z)$ in \eqref{eq:gcp}. Thus, the system poles $\pole_\poleIndex$ are fully characterized by the delays $\Delay$ and the feedback matrix $\Fbm$. \schlecsn{The FDN is stable if all system poles lie within the unit circle. A sufficient stability condition is that the operator norm being $\lVert \Fbm \rVert < 1$ \cite{Rocchesso:1997fv}.}

%The modal decomposition of MIMO FFDN computes the partial fraction decomposition of the transfer function in \eqref{eq:FDNtrans}, i.e., 
%\begin{equation}
%	\TF(z) = \Directgains + \sum_{\poleIndex=1}^\N \frac{\mat{\residue}_\poleIndex}{1 - \pole_\poleIndex \,\inv{z}},
%	\label{eq:modalDecomposition}
%\end{equation}
%where $\mat{\residue}_\poleIndex$ is the matrix of residues of the pole $\pole_\poleIndex$ \cite{Schlecht:2019uj}.

\subsection{Allpass Property}
\label{sec:naive}
A transfer function matrix $\TF(z)$ with real coefficients is allpass if 
\begin{equation}
	\TF(z) \p*{\TF(\inv{z})}\tran = \eye,
	\label{eq:allpass}
\end{equation}
where $\eye$ denotes an identity matrix of appropriate size and $\cdot\tran$ denotes the transpose operation \cite{SmithIII:2007tm}. In particular, $\TF(z)$ is unitary for $z$ on the unit circle. If a MIMO system is allpass then $\det \TF(z)$ is allpass \cite[p.~772]{Vaidyanathan:1993uh}, i.e., 
\begin{equation}
	\abs{\det \TF(\ejw)} \equiv 1 \quad \textrm{ for any } \omega.
	\label{eq:allpassDet}
\end{equation}

%Thus, allpass systems conserve energy. Add Energy balance property in Vaidyanathan.

For allpass filters, the coefficients of the numerator polynomial are in reversed order and possibly with reversed signs of the denominator coefficients \cite{Regalia:1988cj}. Thus, for an allpass FDN in \eqref{eq:FDNtransRational}, there exists $\allpassFactor = \pm 1$ with 
\begin{equation}
	\det \TF(z) = \allpassFactor \frac{z^{-\N} \gcp_{\Delay,\Fbm}(\inv{z}) }{\gcp_{\Delay,\Fbm}(z)}.
	\label{eq:flippedAllpass}
\end{equation}

In the following, we present a classic result for allpass state space systems.

\subsection{Allpass State Space Systems}
For a moment, we consider that all delays are single time steps, i.e., $\Delay = \ones$, where $\ones$ denotes a vector or matrix of ones with appropriate size. The time-domain recursion in \eqref{eq:FDNTimeDomain} reduces to the standard state space realization of a linear time-invariant (LTI) filter. We state a classic sufficient and necessary condition for state space systems to be allpass \cite{Regalia:1988cj,Ferrante:2016rf}.

\begin{theorem}
\label{th:allpassSS}
\schlecsn{Given the $\numIO \times \numIO$ transfer function with realization $\TF(z) = \Outgains \invp{z\eye - \Fbm} \Ingains + \Directgains$. The transfer function $\TF(z)$ is stable and allpass if and only if there exists a symmetric positive definite $\Dsim$ such that }
	\begin{equation}
		\begin{bmatrix}
			\Fbm & \Ingains \\ \Outgains & \Directgains 
		\end{bmatrix}
		\begin{bmatrix}
			\Dsim & \mat{0} \\ \mat{0} & \eye	
		\end{bmatrix}
		\begin{bmatrix}
			\Fbm\tran & \Outgains\tran \\ \Ingains\tran & \Directgains\tran 
		\end{bmatrix}
		=
		\begin{bmatrix}
			\Dsim & \mat{0} \\ \mat{0} & \eye	
		\end{bmatrix}.
		\label{eq:allpassFDNCondition}
	\end{equation}
\end{theorem}

In the Section~\ref{sec:allpassFDN}, we present an extension of this theorem for allpass FDNs.

\subsection{Principal Minors and Diagonal Similarity}
To demonstrate system properties of an FDN independent from delays $\Delay$, we have earlier developed a representation of $\gcp_{\Delay,\Fbm}(z)$ based on the principal minors of $\Fbm$ \cite{Schlecht:2015hi, Schlecht:2017jt}. This representation is also useful to derive the uniallpass property of FDNs.

A principal minor $\det \Fbm(I)$ of a matrix $\Fbm$ is the determinant of a submatrix $\Fbm(I)$ with equal row and column indices $I \subset \Nset$.  The set of all indices is denoted by $\Nset = \{1,2,\dots,N\}$ and $I^c$ is the relative complement in $\Nset$, i.e., $I^c = \Nset \setminus I$. $\abs{I}$ indicates the cardinality of set $I$.

For a given feedback matrix $\Fbm$ and delays $\Delay$, the generalized characteristic polynomial $\gcp_{\Delay,\Fbm}(z)$ is given by
\begin{align}
\label{eq:cPcoefficients}
	\gcp_{\Delay,\Fbm}(z) &= \sum_{k = 0}^{\N} c_k \, z^k   \\
	c_k &= 	
			\begin{cases} 
				\sum_{I \in I_k} (-1)^{N - |I|}\det \Fbm(I^c),  & \text{for } I_k \neq \emptyset\\ 
  				0, & \text{otherwise }
			\end{cases} \nonumber			
\end{align}
where $I_k = \{I \subset \Nset | \sum_{i \in I} \delay_i\ = k\}$. Note that for single sample delays, i.e., $\Delay = \ones$, $\gcp_{\Delay,\Fbm}(z)$ is the standard characteristic polynomial of matrix $\Fbm$. In contrast for $\Delay = [1, 2, \dots, 2^{\matSize-1}]$, $\abs{I_k} = 1$ for $0 \leq k \leq \N$ and therefore each $c_k$ has a single summand in \eqref{eq:cPcoefficients}. Thus, principal minors of $\Fbm$ constitutes a powerful delay-invariant representation.

The principal minors of invertible matrices $\Fbm$ are related by Jacobi's identity \cite{Brualdi:1983ih}, i.e.,
\begin{align}
	\det \inv{\Fbm}(I) = \frac{\det \Fbm(I^c)}{\det \Fbm} \qquad \textrm{ for any } I \subset \Nset.
	\label{eq:jacobi}
\end{align}

\schlecsn{Diagonally similar matrices $\mat{A}$ and $\mat{B}$, i.e., there exists non-singular diagonal matrix $\mat{E}$ with $\mat{E} \mat{A} \inv{\mat{E}} = \mat{B}$, have the same principal minors, however the converse is not true in general \cite{Loewy:1986dg}.}

% however, \schlecsn{if $\inv{\mat{A}}$ and $\mat{A}\tran$ have the same principal minors, then they are diagonally similar \cite[Th.~8]{Schlecht:2017jt}.}

In the following section, we derive the analogue of Theorem~\ref{th:allpassSS} for uniallpass FDNs with arbitrary delays $\Delay$.

%\subsection{Diagonal Similarity}

\section{Uniallpass Feedback Delay Networks}
\label{sec:allpassFDN}
The central question of the present work is which system parameters constitute an allpass transfer function $\TF(z)$ in \eqref{eq:FDNtrans}. In particular, we are interested in uniallpass FDNs, i.e., allpass FDNs with $\Fbm$, $\Ingains$, $\Outgains$, and $\Directgains$ for arbitrary delays $\Delay$. 

%The main result Theorem~\ref{th:allpassFDN}, which is a sufficient condition for uniallpass FDNs is proved at the end of this section. 

\subsection{System Matrix}

First, we establish a convenient notation based on system matrices, i.e.,  
\begin{equation}
	\Sys = \begin{bmatrix}
			\Fbm & \Ingains \\ \Outgains & \Directgains 
			\end{bmatrix}, 
	\label{eq:systemMatrix}
\end{equation}
which is of size $\sysSize \times \sysSize$, where $\sysSize = \matSize + \numIO$.
The Schur complement of the invertible block $\Directgains$ in $\Sys$ is a matrix defined by
\begin{equation}
	\SchurDirect = \Fbm - \Ingains \inv{\Directgains} \Outgains
	\label{eq:SchurDirect}
\end{equation}
and equivalently the Schur complement of the invertible block $\Fbm$ is
\begin{equation}
	\SchurFbm = \Directgains - \Outgains \inv{\Fbm} \Ingains.
\end{equation}

\noindent If $\Fbm$, $\Directgains$, $\SchurDirect$, and $\SchurFbm$ are invertible, the block-wise inverse of the system matrix \eqref{eq:systemMatrix} is
\begin{equation}
	\inv{\Sys} = \begin{bmatrix}
			\invp{\SchurDirect} & -\inv{\Fbm} \Ingains \invp{\SchurFbm} \\ - \invp{\SchurFbm} \Outgains \inv{\Fbm} & \invp{\SchurFbm} 
			\end{bmatrix}.
			\label{eq:inverseSys}
\end{equation}
Further, the inverse of the Schur complements are related by 
\begin{equation}
	\invp{\SchurDirect} = \inv{\Fbm} + 
	\inv{\Fbm}\Ingains \invp{\SchurFbm} \Outgains \inv{\Fbm}.
	\label{eq:SchurComplementLong}
\end{equation}

\schlecsn{
\subsection{Balanced Form}

If there exists a symmetric positive definite $\Dsim$ in \eqref{eq:allpassFDNCondition}, then we can establish a balanced form. There exists a non-singular diagonal matrix $\DiagSim$ with $\Dsim = \DiagSim \DiagSim\tran$ such that by substituting $\FbmS = \inv{\DiagSim} \Fbm \DiagSim$, $\IngainsS = \inv{\DiagSim} \Ingains$, $\OutgainsS = \Outgains \DiagSim$ and, $\DirectgainsS = \Directgains$, we can state \eqref{eq:allpassFDNCondition} as 
\begin{equation}
		\begin{bmatrix}
			\FbmS & \IngainsS \\ \OutgainsS & \DirectgainsS 
		\end{bmatrix}
		\begin{bmatrix}
			\FbmS\tran & \OutgainsS\tran \\ \IngainsS\tran & \DirectgainsS\tran 
		\end{bmatrix}
		= \SysS \SysS\tran = \eye
		.
		\label{eq:allpassFDNTh_balanced}
\end{equation}
As $\Sys$ and $\SysS$ are similar, we have 
\begin{equation}
	\det \Sys = \det \SysS = \pm1.
	\label{eq:sysdetIsOne}
\end{equation}
 
\noindent From Jacobi's identity \eqref{eq:jacobi} with $I_N = \Nset$ in $\langle \sysSize \rangle$, and $\inv{\SysS} = \SysS\tran$, we have
\begin{equation}
\begin{aligned}
	\det \SysS(I_N^c) / \det \SysS &= \det \inv{\SysS}(I_N) \\
	\det \DirectgainsS &= \det \SysS \det \FbmS\tran \\
	\det \Directgains &= \det \Sys \det \Fbm.
	\label{eq:directFeedback}
\end{aligned}
\end{equation}

\subsection{Diagonal Similarity Invariance} 

In the following, we show that the transfer function of an FDN is invariant under diagonal similarity.

\begin{lemma}
\label{lm:diagsim}
Let $\TF(z)$ be an FDN with a realization as in \eqref{eq:FDNtrans}. For any non-singular diagonal matrix $\DiagSim$, we have
	\begin{equation}
		\TF(z) = \TFS(z) = \OutgainsS \invp{ \stdDelayArg{\inv{z}} - \FbmS} \IngainsS + \DirectgainsS,
	\end{equation}	
	where $\FbmS = \inv{\DiagSim} \Fbm \DiagSim$, $\IngainsS = \inv{\DiagSim} \Ingains$, $\OutgainsS = \Outgains \DiagSim$, and, $\DirectgainsS = \Directgains$. We call $\TF(z)$ and $\TFS(z)$ being equivalent.
\end{lemma}

\begin{proof}
By substitution and $\DiagSim \stdDelayArg{\inv{z}} \inv{\DiagSim} = \stdDelayArg{\inv{z}}$, we have
\begin{align*}
	\TFS(z) &= \OutgainsS \invp{ \stdDelayArg{\inv{z}} - \FbmS} \IngainsS + \DirectgainsS \\
			&= \Outgains \DiagSim \invp{ \stdDelayArg{\inv{z}} - \inv{\DiagSim} \Fbm \DiagSim} \inv{\DiagSim} \Ingains + \Directgains \\
			&= \Outgains \invp{ \stdDelayArg{\inv{z}} - \Fbm} \Ingains + \Directgains = \TF(z).
\end{align*}
\end{proof}

As a consequence, we can establish a more refined stability criterion.

\begin{theorem} \label{th:stableFDN}
An FDN realized as in \eqref{eq:FDNtrans} is stable if there exists a non-singular diagonal matrix $\DiagSim$ such that $\anynorm{ \inv{\DiagSim} \Fbm \DiagSim } < 1$.
\end{theorem}

\begin{proof}
An FDN is stable if $\anynorm{\Fbm} < 1$ \cite{Rocchesso:1997fv}. According to Lemma~\ref{lm:diagsim}, for any non-singular diagonal matrix $\DiagSim$, there is an equivalent FDN with feedback matrix $\inv{\DiagSim} \Fbm \DiagSim$. The FDN is therefore stable, if for any such $\DiagSim$, we have $\anynorm{ \inv{\DiagSim} \Fbm \DiagSim } < 1$.
\end{proof}

We can further establish a balanced form for an FDN under diagonal similarity.

\begin{lemma}
	\label{lm:balanced}
If there exists a diagonal positive definite $\Dsim$ with
	\begin{equation}
		\begin{bmatrix}
			\Fbm & \Ingains \\ \Outgains & \Directgains 
		\end{bmatrix}
		\begin{bmatrix}
			\Dsim & \mat{0} \\ \mat{0} & \eye	
		\end{bmatrix}
		\begin{bmatrix}
			\Fbm\tran & \Outgains\tran \\ \Ingains\tran & \Directgains\tran 
		\end{bmatrix}
		=
		\begin{bmatrix}
			\Dsim & \mat{0} \\ \mat{0} & \eye	
		\end{bmatrix},
		\label{eq:DiagLemmaCondition}
	\end{equation}
then there exists an equivalent FDN in balanced form, i.e.,
	\begin{equation}
		\begin{bmatrix}
			\FbmS & \IngainsS \\ \OutgainsS & \DirectgainsS 
		\end{bmatrix}
		\begin{bmatrix}
			\FbmS\tran & \OutgainsS\tran \\ \IngainsS\tran & \DirectgainsS\tran 
		\end{bmatrix}
		=
		\begin{bmatrix}
			\eye & \mat{0} \\ \mat{0} & \eye	
		\end{bmatrix}.
		\label{eq:allpassBalancedCondition}
	\end{equation}
\end{lemma}

\begin{proof}
As $\Dsim$ is diagonal and positive definite, there exists diagonal $\DiagSim$ with $\DiagSim \DiagSim\tran = \Dsim$. According to Lemma~\ref{lm:diagsim}, there is an equivalent FDN with $\FbmS = \inv{\DiagSim} \Fbm \DiagSim$, $\IngainsS = \inv{\DiagSim} \Ingains$, $\OutgainsS = \Outgains \DiagSim$, and, $\DirectgainsS = \Directgains$. By substituting in \eqref{eq:allpassBalancedCondition}, we have
\begin{equation*}
		\begin{bmatrix}
			\inv{\DiagSim} \Fbm \DiagSim & \inv{\DiagSim} \Ingains \\ \Outgains \DiagSim & \Directgains 
		\end{bmatrix}
		\begin{bmatrix}
			\p*{\inv{\DiagSim} \Fbm \DiagSim}\tran & \p*{\Outgains \DiagSim}\tran \\ \p*{\inv{\DiagSim}\Ingains}\tran & \Directgains\tran 
		\end{bmatrix}
		=
		\begin{bmatrix}
			\eye & \mat{0} \\ \mat{0} & \eye	
		\end{bmatrix}.
%		\label{eq:allpassFDNThcondition}
	\end{equation*}
\end{proof}

\subsection{Sufficient Condition for Uniallpass FDNs}
In the following, we derive a sufficient condition for an FDN to be allpass, which is analogous to the sufficient condition of state-space Theorem~\ref{th:allpassSS}. 

\begin{theorem}
\label{th:allpassFDN}
	Given a stable FDN realized as in \eqref{eq:FDNtrans}, then $\TF(z)$ is uniallpass, i.e., allpass for any $\Delay$, if there  exists a diagonal positive definite $\Dsim$ with
	\begin{equation}
		\begin{bmatrix}
			\Fbm & \Ingains \\ \Outgains & \Directgains 
		\end{bmatrix}
		\begin{bmatrix}
			\Dsim & \mat{0} \\ \mat{0} & \eye	
		\end{bmatrix}
		\begin{bmatrix}
			\Fbm\tran & \Outgains\tran \\ \Ingains\tran & \Directgains\tran 
		\end{bmatrix}
		=
		\begin{bmatrix}
			\Dsim & \mat{0} \\ \mat{0} & \eye	
		\end{bmatrix}.
		\label{eq:allpassFDNThcondition}
	\end{equation}
\end{theorem}

\begin{proof}
As the conditions of Lemma~\ref{lm:balanced} are satisfied, we assume that the FDN is in balanced form, i.e., $\Dsim = \eye$. We split \eqref{eq:allpassFDNThcondition} into the individual identities
\begin{equation} \label{eq:balancedSubIdentities}
\begin{aligned}
\Fbm \Fbm\tran + \Ingains \Ingains\tran &= \eye \\
\Fbm \Outgains\tran + \Ingains \Directgains\tran &= \zeros \\
\Outgains \Outgains\tran + \Directgains \Directgains\tran &= \eye .
\end{aligned}
\end{equation}

We show that the allpass condition \eqref{eq:allpass} holds for any $\Delay$. For compactness, we write $\MatPoly(z) = \stdDelayArg{\inv{z}} -  \Fbm$ as in \eqref{eq:loopTransfer} such that
\begin{align*}
	\eye &=  \stdDelayArg{\inv{z}} \stdDelay = \p*{\MatPoly(z) + \Fbm} \p*{\MatPoly(\inv{z}) + \Fbm}\tran \\
	&= \MatPoly(z) \MatPoly(\inv{z})\tran + \MatPoly(z)\Fbm\tran + \Fbm \MatPoly(\inv{z})\tran + \Fbm \Fbm\tran.
\end{align*}
Thus, by using the identities in \eqref{eq:balancedSubIdentities}, we derive
\begin{align*}	
	\zeros &= \MatPoly(z) \MatPoly(\inv{z})\tran + \MatPoly(z)\Fbm\tran + \Fbm \MatPoly(\inv{z})\tran - \Ingains \Ingains\tran \\
	&=  \eye + \Fbm\tran \MatPoly(\inv{z})\itran + \inv{\MatPoly}(z) \Fbm - \inv{\MatPoly}(z) \Ingains \Ingains\tran \MatPoly(\inv{z})\itran \\
	&=  \Outgains \Outgains\tran + \Outgains \Fbm\tran \MatPoly(\inv{z})\itran  \Outgains\tran + \Outgains \inv{\MatPoly}(z) \Fbm  \Outgains\tran \\ &\quad- \Outgains \inv{\MatPoly}(z) \Ingains \Ingains\tran \MatPoly(\inv{z})\itran  \Outgains\tran \\
	&= \eye - \Directgains \Directgains\tran - \Directgains \Ingains\tran \MatPoly(\inv{z})\itran  \Outgains\tran - \Outgains \inv{\MatPoly}(z) \Ingains \Directgains\tran \\ &\quad- \Outgains \inv{\MatPoly}(z) \Ingains \Ingains\tran \MatPoly(\inv{z})\itran  \Outgains\tran.
\end{align*}
Thus,
\begin{align}
	\eye &= \p*{\Directgains + \Outgains \inv{\MatPoly}(z) \Ingains } \p*{\Directgains + \Outgains \inv{\MatPoly}(\inv{z}) \Ingains }\tran   \\
	&= \TF(z) \p*{\TF(\inv{z})}\tran.
\end{align}
Therefore for any $\Delay$, the transfer function $\TF(z)$ is allpass.

%Given a diagonal matrix $\Dsim > 0$ satisfying \eqref{eq:allpassFDNThcondition}, due to Lemma~\ref{lm:similarity}, $\Directgains$ is non-singular as $\Fbm$ is non-singular and 
%\begin{equation}
%	\det \SchurDirectS(I) = \det \Sys \det \FbmS(I^c) / \det \Directgains.
%\end{equation}
%And as $\Dsim$ is diagonal, we have $\det \SchurDirectS(I) = \det \SchurDirect(I)$ and $\FbmS(I^c) = \Fbm(I^c)$. Thus, the FDN satisfies \eqref{eq:allpassPMs} and according to Lemma~\ref{lm:PM}, the FDN is uniallpass.

%Please note that the diagonal similarity transform does not alter the transfer function $\TF(z)$.
%
% Therefore, $\TF(z)$ is allpass.
%	
%	as in 
	
%	\schlecsn{Alternatively use normalized form and use energy balance}

\end{proof}

For such a uniallpass FDN, we have $\det \Directgains = \pm \det \Fbm$, see \eqref{eq:directFeedback}. Thus, like in Schroeder allpass structures \cite{Schroeder:1961ke}, there is an inherent relation between the direct component and the decay rate of the response.

\subsection{Necessary Condition for Uniallpass FDNs}
The main challenge in the following theorem is that the allpass property is to be independent of the choice of the delays $\Delay$. Therefore, we give a necessary condition based on the principal minors of the system matrix $\Sys$ alone.

\begin{theorem} \label{th:PM}
If a stable FDN realized as in \eqref{eq:FDNtrans} with non-singular $\Directgains$ is uniallpass, then there exists $\allpassFactor = \pm 1$ with
	\begin{equation}
	\det \SchurDirect(I) = \allpassFactor \det \inv{\Fbm}(I) \qquad \forall I \subset \Nset.
	\label{eq:allpassPMs}
	\end{equation}
For the SISO case, i.e., $\numIO = 1$, the FDN is uniallpass if and only if \eqref{eq:allpassPMs} holds.
\end{theorem}

\begin{proof}
If the FDN is stable and uniallpass, then it is also allpass for $\Delay = \ones$. Therefore, Theorem~\ref{th:allpassSS} applies and due to \eqref{eq:sysdetIsOne} and \eqref{eq:directFeedback}, we have $\det \Sys = \pm 1$ and $\Fbm$ is non-singular if and only if $\Directgains$ is non-singular.

According to \eqref{eq:allpassDet}, if the FDN is allpass then the determinant of the transfer function $\det \TF(z)$ is allpass. Applying the matrix determinant lemma \cite{Petersen:2012up} in \eqref{eq:FDNtrans} and using the Schur complement notation \eqref{eq:SchurDirect}, we have  
\begin{align}
	\det \TF(z) &= \frac{ \detp{ \stdDelayArg{\inv{z}} - \Fbm + \Ingains \inv{\Directgains} \Outgains } \det \Directgains  }{ \detp{  \stdDelayArg{\inv{z}} - \Fbm } } \\
	&= \frac{ \gcp_{\Delay,\SchurDirect}(z) \det \Directgains  }{ \gcp_{\Delay,\Fbm}(z) }.
	\label{eq:TFdeterminant}
\end{align} 
According to \eqref{eq:flippedAllpass}, for $\det \TF(z)$ to be allpass, the coefficients of denominator and numerator of \eqref{eq:TFdeterminant} are in reversed order, i.e., there exists $\allpassFactor = \pm 1$ such that
\begin{equation}
	\gcp_{\Delay,\SchurDirect}(z) \det \Directgains =  \allpassFactor z^{-\N} \gcp_{\Delay,\Fbm}(\inv{z}). 
	\label{eq:GCPreversed}
\end{equation}

\noindent For the special case $\Delay = [1, 2, \dots, 2^{\matSize-1}]$, \eqref{eq:GCPreversed} holds if and only if 
\begin{equation*}
	\det \Directgains \det \SchurDirect(I) = \allpassFactor \det \Fbm(I^c) \qquad \forall I \subset \Nset
	%\label{eq:allpassPMs}
\end{equation*}
as $\abs{I_k} = 1$ for any $k$ such that each coefficient $c_k$ in \eqref{eq:cPcoefficients} has a single summand. Applying Jacobi's identity \eqref{eq:jacobi} and \eqref{eq:directFeedback} yields \eqref{eq:allpassPMs}. For the SISO case, \eqref{eq:allpassDet} is also a sufficient condition for the FDN to be allpass.
\end{proof}

To develop a necessary condition based on the diagonal similarity of the system matrix as in \eqref{eq:allpassFDNThcondition}, likely additional constraints are required. For instance with an additional rank condition on $\Fbm$, the correspondence of the principal minors \eqref{eq:allpassPMs} yields a diagonal similarity between $\Fbm$ and $\SchurDirect$ \cite{Loewy:1986dg}.

\subsection{Allpass is not Uniallpass} \label{sec:allpassNotUniallpass}

We use Theorem~\ref{th:PM} to construct an example for an FDN which is allpass for only certain delays $\Delay$, but not for other choices and therefore not being uniallpass:
\begin{align*}
	\Fbm &= \begin{bmatrix} 
1.241 & 3.833 &  -6.028 \\ 
-0.859 & -2.276 &  3.582 \\ 
-0.048 & -0.180 &  -0.332 \\ 
\end{bmatrix},\\
	\Ingain\tran &= \begin{bmatrix} 
1.833 & -0.469 &  0.826 \\ 
\end{bmatrix},  \\
	\Outgain &= \begin{bmatrix} 
0.430 & 0.831 &  0.452 \\ 
\end{bmatrix},  \\
	\directgain &= 0.288.
\end{align*}
The principal minors of $\inv{\Fbm}$ and $\SchurDirect$ are, respectively, 
\begin{align*}
	&\begin{bmatrix} 
		1.00 , -4.86 , \hphantom{-}2.44 , -1.63 , 1.15 ,\hphantom{-}7.89 , -4.30 ,  -3.47 \\ 
	\end{bmatrix},\\
	&\begin{bmatrix} 
		1.00 , -1.49 , -0.92 , -1.63 , 1.15 , -8.97 , 12.56 ,  -3.47 \\
	\end{bmatrix}.
\end{align*}
The FDN is not uniallpass as only some of the principal minors coincide. However, the FDN is allpass for $\Delay = [1,1,1]$ as the polynomial coefficients of the transfer function numerator and denominator in \eqref{eq:FDNtransRational}, respectively, are in reverse order, see \eqref{eq:flippedAllpass}:
\begin{align*}
	&\begin{bmatrix} 
		0.29 & 1.17 & 1.37 &  1.00 \\ 
	\end{bmatrix} \\
	&\begin{bmatrix} 
		1.00 & 1.37 & 1.17 &  0.29 \\ 
	\end{bmatrix}.
\end{align*}
However, the FDN is not allpass for $\Delay = [2, 1, 1]$ as the numerator and denominator, respectively, are
\begin{align*}
	&\begin{bmatrix} 
0.29 & 0.74 & 4.05 & -2.26 &  1.00 \\ 
\end{bmatrix}\\
&\begin{bmatrix} 
1.00 & 2.61 & 0.16 & -0.23 &  0.29 \\ 
\end{bmatrix}.
\end{align*}
Then again, the FDN is allpass for $\Delay = [2, 2, 1]$ as the numerator and denominator, respectively, are
\begin{align*}
&\begin{bmatrix} 
0.29 & 0.47 & 0.70 & 1.03 & 0.33 &  1.00 \\ 
\end{bmatrix}\\
&\begin{bmatrix} 
1.00 & 0.33 & 1.03 & 0.70 & 0.47 &  0.29 \\ 
\end{bmatrix}.
\end{align*}
This example illustrates that for non-uniallpass FDNs, the allpass property intricately depends on $\Delay$. For larger systems, it becomes increasingly complex to determine the allpass property. Uniallpass FDNs provide an alternative, where the delays are an independent design parameter.

In the following section, we present methods to design uniallpass FDNs based on a desired feedback matrix $\Fbm$.   

%While in \eqref{eq:allpassFDNThcondition}, $\Dsim$ is diagonal for uniallpass FDN, $\Dsim$ is not necessarily diagonal for specific delays $\Delay$. For instance, allpass FDNs with equal delays $\Delay = k\ones$ with $k \in \set{N}$, $\Dsim$ is only necessarily symmetric as in Th.~\ref{th:allpassSS}. For longer delays $\Delay$, it can become quickly impractical to determine the allpass property for specific $\Delay$ such that the uniallpass property is more useful albeit slightly restrictive. 
%However, based on observations of unilossless matrices, we conjecture for many $\Delay$ that $\Dsim$ tends to be close to diagonal \cite{Schlecht:2017jt}.
}

\section{Uniallpass FDN Completion}
\label{sec:completion}

Uniallpass FDNs can be generated by a simple procedure for $\numIO$ input and output channels and $\matSize$ delay lines. First, generate an orthogonal system matrix $\Sys$ of size $\sysSize \times \sysSize$ with $\sysSize = \systemOrder + \numIO$. Optionally, apply a similarity transform with a non-singular diagonal matrix $\diag{\Dsim, \eye}$. However, note that the similarity transform does not alter the transfer function, but may change computational properties. Lastly, divide the system matrix $\Sys$ into the submatrices $\Fbm$, $\Ingains$, $\Outgains$, and $\Directgains$ according to \eqref{eq:systemMatrix}. However, this procedure does not allow to specify directly the feedback matrix $\Fbm$ and the resulting filter properties.

In this section, we present procedures related to the completion problem, i.e., determining $\Ingains$, $\Outgains$, and $\Directgains$ given $\Fbm$ such that $\Sys$ is uniallpass. The following subsections are: \ref{sec:determiningDiagonal}) determining $\Dsim$ given uniallpass $\Sys$; \ref{sec:admissible}) characterize admissible feedback matrices $\Fbm$; \ref{sec:orthogonalCompletion}) completion where $\Dsim = \eye$; and, \ref{sec:generalCompletion}) completion for any diagonal $\Dsim$.

\subsection{Determining Diagonal Similarity}
\label{sec:determiningDiagonal}
Given a uniallpass FDN as in Theorem~\ref{th:allpassFDN} with system matrix $\Sys$, the diagonal similarity matrix $\Dsim$ in \eqref{eq:allpassFDNThcondition} can be computed by solving the discrete-time Lyapunov equation\schlecsn{\footnote{\schlecsn{The discrete-time Lyapunov equation has off-the-shelf solver implementations such as dlyap in MATLAB.}}} \cite{Kaszkurewicz:2000}
\begin{equation}
	\Dsim - \Fbm \Dsim \Fbm\tran = \Ingains \Ingains\tran.
	\label{eq:lyapunovB}
\end{equation}
We give an alternative solution, which is helpful for the further development below. The system matrix $\Sys$ satisfies \eqref{eq:allpassFDNThcondition}, thus $\Sys$ is diagonally similar to an orthogonal matrix. We review here, key aspects of Engel and Schneider's algorithm to determine the diagonal similarity \cite{Engel:1982}. 

%To reduce the technical details of the procedure, here we assume that the underlying graph is fully connected, i.e., $\Sys$ has no zero entries.    
A system matrix $\Sys$ is diagonally similar to an orthogonal matrix if and only if $\inv{\Sys} \hadaQuot \Sys\tran$ is diagonally similar to a $\{0,1\}$-matrix $\ZeroOne$, i.e., $\ZeroOne \in \{0,1\}^{\sysSize \times \sysSize}$ \cite[Corollary 4.7 and 3.11]{Engel:1982}. Operation $\hadaQuot$ denotes an element-wise division also called Hadamard quotient, i.e.,
\begin{equation}
	\p*{\mat{A} \hadaQuot \mat{B}}_{ij} = \begin{cases}
		a_{ij} / b_{ij} \qquad &\textrm{ for }  b_{ij} \neq 0 \\
		0 \qquad &\textrm{ otherwise.}
	\end{cases}
\end{equation}

Thus with \eqref{eq:inverseSys}, the similarity transform $\Dsim$ can be readily retrieved from 
\begin{equation}
	\inv{\Dsim} \ZeroOne \Dsim = \invp{\SchurDirect} \hadaQuot \Fbm \tran.
	\label{eq:diagSimHadamard}
\end{equation}
For fully connected matrices $\Fbm$ and $\invp{\SchurDirect}$, i.e., having only non-zero elements, $\ZeroOne$ contains only ones. Then, \eqref{eq:diagSimHadamard} can be simply solved by a singular value decomposition. For non-fully connected $\Fbm$ and $\invp{\SchurDirect}$, the computation is performed on the spanning tree of the adjacency graph of $\Fbm$, for more details see \cite{Engel:1982}.

\subsection{Admissible Feedback Matrix}
\label{sec:admissible}
In the following, we characterize the feedback matrix $\Fbm$ of uniallpass FDNs with system matrix $\Sys$. First, we assume that $\Sys$ is orthogonal. The following theorem by Fiedler \cite{Fiedler:2009} gives sufficient and necessary conditions for such $\Fbm$.

\begin{theorem}[Fiedler \cite{Fiedler:2009}, Theorem 2.2]
\label{th:Fiedler}
Every $\matSize \times \matSize$ submatrix of an orthogonal $\sysSize \times \sysSize$ matrix has at least $2\matSize - \sysSize = \matSize - \numIO$ singular values equal to one and $\numIO$ singular values less than one.
 
Conversely, if $\Fbm$ is a $\matSize \times \matSize$ matrix that has $\matSize - k$ singular values equal to one and the remaining $k$ singular values less than one, then for every $\sysSize \geq \matSize + k$ there exists an orthogonal $\sysSize \times \sysSize$ matrix containing $\Fbm$ as a submatrix, and for no $\sysSize$ smaller than $\matSize + k$ does such matrix exist.
\end{theorem}

In particular for the SISO case with $\numIO = 1$, $\Fbm$ has exactly one singular value less than one and the other singular values are one. In the full MIMO case, i.e., $\numIO = \matSize$, $\Fbm$ has all singular values less than one. Thus, any admissible feedback matrix $\Fbm$ of a uniallpass FDN is diagonally similar to a matrix with singular values as described above. There are various techniques to generate matrices with prescribed eigenvalues and singular values \cite{Chu:2000ez, Li:2001be}. Note, that for a stable FDN, the moduli of the eigenvalues of $\Fbm$ are less than one \cite{Rocchesso:1997fv}.

\subsection{Orthogonal Completion}
\label{sec:orthogonalCompletion}
We give a simple method for completing an orthogonal uniallpass system. Given an $\matSize \times \matSize$ submatrix $\Fbm$ of an $\sysSize \times \sysSize$ orthogonal matrix $\Sys$, i.e., $\Sys \Sys\tran = \eye$. Therefore, $\Dsim = \eye$ in \eqref{eq:allpassFDNThcondition}. The block matrices in \eqref{eq:allpassFDNThcondition} for $\Sys \Sys\tran = \eye$ and $\Sys\tran \Sys = \eye$ yield then
\begin{align}
	\eye - \Fbm \Fbm\tran &= \Ingains \Ingains\tran,\\
	\eye - \Fbm\tran \Fbm &= \Outgains\tran \Outgains,\\
	-\Ingains \Directgains\tran &= \Fbm \Outgains\tran.
\end{align}
The equations can be solved with a singular value decomposition, e.g., $\Ingains$ is the rank-$\numIO$ decomposition of $\eye - \Fbm \Fbm\tran$.

Particularly in the full MIMO case, any matrix $\Fbm$ with all singular values less than one can be completed to a uniallpass FDN. As demonstrated in the Section~\ref{sec:application}, this result is an extension to prior designs.

%\schlecsn{degrees of freedom}

%\subsection{Feedback Matrix Completion}
%
%The condition above is equal to the matrix equations (first one of Lyapunov type):
%\begin{align}
%	\Fbm \Dsim \Fbm\tran + \Ingains \Dsim \Ingains\tran = \Dsim \\
%	\Fbm \Dsim \Outgains + \Ingains \Dsim \Directgains\tran = \zeros \\
%	\Outgains\tran \Outgains + \Directgains \Directgains\tran = \eye
%	\label{eq:lyapunov}
%\end{align}
%
%
%
%
%
%
%\subsection{Orthogonal Case - other term (balanced?)}

\subsection{General Completion}
\label{sec:generalCompletion}

%If we assume know that $\Fbm$ is part of any allpass FDN. Thus, there is a diagonal matrix $\DsimQ$ such that 
%\begin{equation}
%	\Sim{\Sys} = \begin{bmatrix}
%		 \Sim{\Fbm} & \Sim{\Ingains} \\ \Sim{\Outgains}\tran & \Sim{\Directgains}
%	\end{bmatrix}
%\end{equation}
%is orthogonal where $\Sim{\Fbm} = \inv{\DsimQ} \Fbm \DsimQ$, $\Sim{\Ingains} = \inv{\DsimQ} \Ingains $, $\Sim{\Outgains} = \DsimQ \Outgains$, $\Sim{\Directgains} = \Directgains$.
Here, we complete a feedback matrix $\Fbm$, which is part of any (not necessarily orthogonal) uniallpass FDN. The first part of the procedure is general, whereas the latter part focuses on the SISO case. From \eqref{eq:allpassFDNThcondition} and \eqref{eq:inverseSys}, we have 
\begin{equation}
	\Directgains\tran = \invp{\SchurFbm} = \invp{\Directgains - \Outgains \inv{\Fbm} \Ingains}
\end{equation}
and further
\begin{align}
	-\inv{\Fbm} \Ingains \Directgains\tran &= \Dsim \Outgains\tran, \label{eq:recoverP} \\
	- \Directgains\tran \Outgains \inv{\Fbm}   &= \Ingains\tran \inv{\Dsim}.
\end{align}
Therefore, \eqref{eq:SchurComplementLong} is
\begin{equation}
	\invp{\SchurDirect} = \inv{\Fbm} + 
	\Dsim \Outgains\tran {\Directgains}\itran \Ingains\tran \inv{\Dsim}.
\end{equation}
Given the system matrix $\Sys$ of a uniallpass FDN, thus, $\Sys\tran$ and $\inv{\Sys}$ are diagonally similar and the Hadamard quotient $\Sys\tran \hadaQuot \inv{\Sys}$ is diagonally similar to a $\{0,1\}$-matrix. Thus, 
\begin{equation}
	\Quot = \p*{\inv{\Fbm} + 
	\Dsim \Outgains\tran {\Directgains}\itran \Ingains\tran \inv{\Dsim}} \hadaQuot \Fbm\tran 
	\label{eq:diagSimInverse}
\end{equation}
is diagonally similar to a $\{0,1\}$-matrix $\ZeroOne$. In particular, the diagonal elements of $\Quot$ are ones, and therefore 
\begin{equation}
	\p*{\Fbm}_{ii}  =  \p*{\inv{\Fbm}}_{ii} + \p*{\Outgains\tran {\Directgains}\itran \Ingains\tran}_{ii}.
	\label{eq:diagCondition}
\end{equation}

The remaining procedure is only for the SISO case, which is emphasized by using vectors and scalars $\Ingain$, $\Outgain$ and $\directgain$ instead of matrices. From the uniallpass property, we have $\directgain = \pm \det \Fbm$. We restate \eqref{eq:diagSimInverse}
\begin{equation}
	\Quot =  \p*{\inv{\Fbm} + 
	\frac{\Sim{\Outgain}\tran \Sim{\Ingain}\tran}{\directgain}} \hadaQuot \Fbm\tran ,
	\label{eq:diagSimInverseSISO}
\end{equation}
where $\Sim{\Outgain} = \Outgain \Dsim$, ${\Sim{\Ingain}} = \inv{\Dsim} {\Ingain}$. We can also rewrite \eqref{eq:diagCondition} for the SISO case, i.e., 
\begin{equation}
	\p*{\Fbm}_{ii} = \p*{\inv{\Fbm}}_{ii} + \p*{\Outgain\tran \inv{\directgain} \Ingain\tran}_{ii}.
\end{equation} 
More concisely, we can write
\begin{equation}
	\directgain \FbmDiag = \Outgain\tran \hadamard \Ingain =  \Sim{\Outgain}\tran \hadamard \Sim{\Ingain},
\end{equation} 
where 
\begin{equation} \label{eq:fbmDiag}
	\fbmDiag_{i} = \p*{\Fbm}_{ii} - \p*{\inv{\Fbm}}_{ii}
\end{equation}
and $\hadamard$ denotes the element-wise product, also called Hadamard product. By inspecting the individual matrix entries for $1 \leq i,j \leq \matSize$
\begin{equation}
	\p*{\Sim{\outgain}_i \Sim{\ingain}_j} \p*{\Sim{\ingain}_i \Sim{\outgain}_j} = \Sim{\outgain}_i \Sim{\outgain}_j \Sim{\ingain}_i \Sim{\ingain}_j = \p*{\Sim{\outgain}_i \Sim{\ingain}_i} \p*{\Sim{\outgain}_j \Sim{\ingain}_j},
\end{equation}
we derive an important identity
\begin{equation}
	\frac{\Sim{\Outgain}\tran \Sim{\Ingain}\tran}{\directgain} \hadamard \frac{\Sim{\Ingain} \Sim{\Outgain}}{\directgain} = \frac{( \Sim{\Outgain}\tran \hadamard \Sim{\Ingain} ) ( \Sim{\Outgain}\tran \hadamard \Sim{\Ingain} )\tran}{\directgain^2} = \FbmDiag \FbmDiag\tran.  
	\label{eq:cbcb}
\end{equation}

Because $\Quot$ is diagonally similar to a $\{0,1\}$-matrix $\ZeroOne$, we have 
\begin{equation}
	\Quot \hadamard \Quot\tran = \ZeroOne.
	\label{eq:quotQuot}
\end{equation}
We use this identity in the following to determine the input and output gains. By substituting  \eqref{eq:diagSimInverseSISO} and \eqref{eq:cbcb} in $\Quot \hadamard \Quot\tran$, we derive
\begin{equation}
\begin{split}
	\Quot \hadamard \Quot\tran \hadamard \Fbm \hadamard \Fbm\tran &= \inv{\Fbm} \hadamard \Fbm\itran + \\ &
	\inv{\Fbm} \hadamard \frac{\Sim{\Ingain} \Sim{\Outgain}}{\directgain} + \p*{\inv{\Fbm} \hadamard \frac{\Sim{\Ingain} \Sim{\Outgain}}{\directgain}}\tran +  \FbmDiag \FbmDiag\tran .
\end{split}
\label{eq:quotQuot2}
\end{equation}
\schlecsn{Because of \eqref{eq:diagSimInverseSISO}, $\Quot \hadamard \Quot\tran \hadamard \Fbm \hadamard \Fbm\tran = \Fbm \hadamard \Fbm\tran$ and can be simplified in \eqref{eq:quotQuot2}.}
By substituting \eqref{eq:quotQuot} into \eqref{eq:quotQuot2} and by sorting the terms we can write more concisely, 
\begin{equation}
	\inv{\Fbm} \hadamard \Sim{\Ingain} \Sim{\Outgain} + {\Fbm}\itran \hadamard \Sim{\Outgain}\tran \Sim{\Ingain}\tran = \rightHand,
\end{equation}
where 
\begin{equation}
	\rightHand = \directgain (\Fbm \hadamard \Fbm\tran - \inv{\Fbm} \hadamard \Fbm\itran -  \FbmDiag \FbmDiag\tran).
	\label{eq:rightHand}
\end{equation}
%Note that for orthogonal $\Fbm$, we have $\Fbm \hadamard \Fbm\tran - \inv{\Fbm} \hadamard \Fbm\itran = \zeros$. \schlecsn{Can be used for homogeneous proof.}

By Hadamard multiplying the equation with $\Sim{\Ingain} \Sim{\Outgain}$ and substituting \eqref{eq:cbcb}, we get
\begin{equation}
	\inv{\Fbm} \hadamard \p*{\Sim{\Ingain} \Sim{\Outgain}}^{\hadamard 2} - \rightHand \hadamard \Sim{\Ingain} \Sim{\Outgain} + {\Fbm}\itran \hadamard \directgain^2 \FbmDiag \FbmDiag\tran  = \zeros,
	\label{eq:quadraticEquation}
\end{equation}
where $\cdot^{\hadamard 2}$ denotes the element-wise square. Each matrix entry in \eqref{eq:quadraticEquation} is a quadratic equation and can be solved independently. From the two possible solutions for each matrix entry, one is selected such that the solution matrix is of rank 1. From \eqref{eq:recoverP}, 
\begin{equation}
	- \Dsim \Sim{\Ingain} \directgain = - \Ingain \directgain = \Fbm \Sim{\Outgain}\tran
	\label{eq:recoverIngain}
\end{equation}
such that
\begin{equation}
	\diag \Dsim  = - \p*{ \Fbm \Sim{\Outgain}\tran } \hadaQuot \p*{\Sim{\Ingain} \directgain}
	\label{eq:recoverDsim}
\end{equation}
we can recover $\Dsim$ and therefore $\Ingain$ and $\Outgain$ from $\Sim{\Ingain}$ and $\Sim{\Outgain}$. This concludes the completion algorithms for SISO uniallpass FDNs. \schlecsn{We summarize the computational steps: 
\begin{enumerate}
	\item Compute $\directgain = \pm \det \Fbm$, see \eqref{eq:directFeedback}.
	\item Compute $\FbmDiag$, see \eqref{eq:fbmDiag}.
	\item Compute $\rightHand$, see \eqref{eq:rightHand}.
	\item Solve for $\Sim{\Ingain} \Sim{\Outgain}$, see \eqref{eq:quadraticEquation}.
	\item Compute $\Dsim$, see \eqref{eq:recoverDsim}.
	\item Compute $\Ingain$ and $\Outgain$, see	below \eqref{eq:diagSimInverseSISO}.
\end{enumerate}
A MATLAB implementation is provided in the FDN toolbox \cite{Schlecht.2020.FDNTB}.} In the following section, we study the completion of a special class of feedback matrices.

\section{Homogeneous Decay Allpass FDN}
\label{sec:homogeneous}

\subsection{Homogeneous Decay}
 A typical requirement in artificial reverberation and audio decorrelation is that all modes decay at the same rate, i.e., all system eigenvalues have the same magnitude, i.e., $\abs{\pole_i} = \decayRate$ for $1 \leq i \leq \N$. We refer to this property as homogeneous decay. In FDNs, this can be achieved by delay-proportional absorption in combination with a lossless matrix \cite{Jot:1991tq}. Thus, the feedback matrix is
\begin{equation}
	\Fbm = \Unitary \Decay
	\label{eq:homoMatrix}
\end{equation}
with unilossless matrix $\Unitary$, diagonal matrix $\Decay$ with \cite{Schlecht:2017jt}
\begin{equation}
	\decay_{ii} = \decayRate^{\delay_i} \textrm{ for } 1 \leq i \leq \matSize.
	\label{eq:decayGains} 
\end{equation}
For $0 < \decayRate < 1$, the singular values of $\Fbm$ are then $\decay_{11}, \dots, \decay_{\matSize\matSize}$ and the eigenvalues of $\Fbm$ have moduli less than 1. From Section~\ref{sec:orthogonalCompletion}, any such feedback matrix can be completed into a full MIMO uniallpass FDN. Note that this is a significant extension to Poletti's design \cite{Poletti:1995tq} as shown below in Section~\ref{sec:application}. In \eqref{eq:homoMatrix}, $\Unitary$ can be a unilossless triangular matrix, i.e., with a diagonal of ones \cite{Schlecht:2017jt}. In Section~\ref{sec:application}, we revisit this structure for series allpasses. In the following, we focus on the more intricate case of orthogonal $\Unitary$.

\subsection{SISO FDN}
\label{sec:homogeneousSISO}
\schlecsn{We construct homogeneous decay uniallpass FDNs for SISO, i.e., $0 < \decayRate < 1$ in \eqref{eq:decayGains}}. We substitute \eqref{eq:homoMatrix} into \eqref{eq:lyapunovB},
\begin{equation}
	\Dsim - \Unitary \Decay \Dsim \Decay \Unitary\tran = \Ingain \Ingain\tran.
\end{equation}
We right-multiply with $\Unitary$ and substitute $\DsimQ = \Decay^2 \Dsim$ and $\hat{ \Ingain }  = \Unitary\tran \Ingain$ such that
\begin{equation}
	\Dsim \Unitary - \Unitary \DsimQ = \Ingain \IngainU\tran,
	\label{eq:displacement}
\end{equation}
which is called a displacement equation \cite{Fasino:2002}. In the following, we denote the diagonal entries of a diagonal matrix $\Dsim$ with a single index, e.g., $\dsim_{ii} = \dsim_i$. The solution of the displacement equation \eqref{eq:displacement} is the Cauchy-like matrix \cite{Fasino:2002}
\begin{equation}
\begin{aligned}
	\Unitary &= \Ingain \IngainU\tran \hadamard \Cauchy \\
		&= \diag{\Ingain} \Cauchy \diag{\IngainU},
	\label{eq:CauchylikeCondition}
\end{aligned}
\end{equation}
where the $\matSize \times \matSize$ Cauchy matrix $\Cauchy$ has elements
\begin{equation}
	\cauchy_{ij} = \frac{1}{\dsim_i - \dsimQ_j }.
\end{equation}
Then, the inverse of the Cauchy matrix is given by \cite{Schechter:1959}
\begin{equation}
	\inv{\Cauchy} = \diag{\vec{\alpha}} \Cauchy\tran \diag{\vec{\beta}},
	\label{eq:inverseCauchy}
\end{equation}
where the elements of $\matSize \times 1$ vectors $\vec{\alpha}$ and $\vec{\beta}$ are
\begin{equation}
	\alpha_i = -\frac{\mathcal{A}(\dsimQ_i) }{ \mathcal{B}\der(\dsimQ_i) } \textrm{ and }
	\beta_i = \frac{\mathcal{B}(\dsim_i) }{ \mathcal{A}\der(\dsim_i) }
	\label{eq:alphaBeta}
\end{equation}
and
\begin{equation}
	\mathcal{A}(x) = \prod_{k = 1}^\matSize (x - \dsim_k)  \textrm{ and }
	\mathcal{B}(x) = \prod_{k = 1}^\matSize (x - \dsimQ_k) ,
	\label{eq:ABpoly}
\end{equation}
where $\cdot\der$ denotes the derivative with respect to $x$. Thus, the diagonal elements of $\Dsim$ and $\DsimQ$ are the zeros of the polynomials $\mathcal{A}(x)$ and $\mathcal{B}(x)$. Thus, taking the inverse in \eqref{eq:CauchylikeCondition} and substituting \eqref{eq:inverseCauchy}, yields
\begin{equation}
\begin{aligned}
	\inv{\Unitary} &= \inv{\diag{\IngainU}} \inv{\Cauchy} \inv{\diag{\Ingain}} \\	
		&=  \inv{\diag{\IngainU}} \diag{\vec{\alpha}} \Cauchy\tran \diag{\vec{\beta}} \inv{\diag{\Ingain}}.
\end{aligned}
\end{equation}
Because $\Unitary\tran = \inv{\Unitary}$, we have
\begin{equation}
	\diag{\IngainU}^2 = \diag{\vec{\alpha}} \textrm{ and } \diag{\Ingain}^2 = \diag{\vec{\beta}}.
	\label{eq:admissiableAlphaBeta}
\end{equation}
Therefore, $\vec{\alpha}$ and $\vec{\beta}$ need to be positive. And the unitary matrix is given by 
\begin{equation}
	\unitary_{ij} = \frac{\sqrt{\beta_i \alpha_j}}{ \dsim_i - \dsimQ_j}.
	\label{eq:unitaryAlpha}
\end{equation}

\subsection{Admissible Parameters}
 Firstly, we give a sufficient condition for $\Dsim$ and $\DsimQ$ to be admissible, i.e., $\vec{\alpha}$ and $\vec{\beta}$ in \eqref{eq:admissiableAlphaBeta} are positive. Secondly, for a given decay gains $\Decay$, we determine similarity matrix $\Dsim$ such that $\Dsim$ and $\DsimQ = \Decay^2 \Dsim$ are admissible. The choice of $\Dsim$ is effectively a parametrization of $\Unitary$ in \eqref{eq:unitaryAlpha} such that a uniallpass FDN exists with $\Fbm = \Unitary \Decay$.
 
%Without loss of generality, let us assume that
%\begin{equation}
%	\dsim_1 < \dsim_2 < \dots < \dsim_\matSize \textrm{ and } \dsimQ_1 < \dsimQ_2 < \dots < \dsimQ_\matSize
%\end{equation} 

We show that the following choice of $\Dsim$ and $\DsimQ$ is admissible, i.e., 
\begin{equation}
	\dsimQ_1 < \dsim_1 < \dsimQ_2 < \dsim_2 < \dots < \dsimQ_\matSize < \dsim_\matSize.
	\label{eq:interleavedP}
\end{equation}
Because of \eqref{eq:ABpoly}, we say that the zeros of $\mathcal{A}(x)$ and $\mathcal{B}(x)$ are strictly interlaced. 

With Rolle's theorem, the zeros of the derivatives $\mathcal{A}\der(x)$ and $\mathcal{B}\der(x)$ are strictly interleaving the zeros of $\mathcal{A}(x)$ and $\mathcal{B}(x)$, respectively \cite{Rahman:2002tq}. Thus, with \eqref{eq:interleavedP}, we have that 
\begin{equation}
	\sign \mathcal{A}'(\dsim_i) = \sign \mathcal{B}'(\dsimQ_i) = (-1)^{\matSize - i},
\end{equation}
where $\sign$ denotes the sign operator. Similarly, because of \eqref{eq:interleavedP}, we have 
\begin{equation}
	\sign \mathcal{A}(\dsimQ_i) = (-1)^{\matSize+1-i} \textrm{ and } \sign \mathcal{B}(\dsim_i) = (-1)^{\matSize-i}.
\end{equation}
Therefore, with \eqref{eq:alphaBeta}, we have
\begin{equation*}
	\sign \alpha_i = -\frac{(-1)^{\matSize+1-i} }{ (-1)^{\matSize - i} } = 1 \textrm{ and }
	\sign \beta_i = \frac{(-1)^{\matSize-i} }{ (-1)^{\matSize - i} } = 1
\end{equation*}
such that $\Dsim$ and $\DsimQ$ in \eqref{eq:interleavedP} yield an admissible solution to \eqref{eq:admissiableAlphaBeta}.

%\todo{not clear, motivate at the beginning}
Thus, for a given decay gain $\Decay$, we choose $\Dsim$ such that $\Dsim$ strictly interleaves $\DsimQ = \Dsim \Decay^2$. With \eqref{eq:interleavedP}, we have
\begin{equation}
	0 < {\frac{\dsim_{i-1}}{\dsim_{i}}} < \decay_{i}^2 \textrm{ for } 2 \leq i \leq \matSize 
	\label{eq:homoChoice}
\end{equation}
and $\dsim_1$ and $\decay_1 < 1$ are unconstrained. Note, that $\Decay$ does not need to be sorted in any way. As we have not constrained the decay gains $\Decay$, we have shown that there exists SISO uniallpass FDNs with homogeneous decays for any delay $\Delay$ and any decay rate $0 < \decayRate < 1$. The similarity matrix $\Dsim$ acts as an additional design parameter within the constraints of \eqref{eq:homoChoice}.

%\schlecsn{discuss with Fiedler singular values}

%%%%%%%%%%%%%%%%%%%%%%%
\begin{figure}[!tb]
\centering
\begin{subfigure}[b]{\columnwidth}
	\includegraphics[width=\columnwidth]{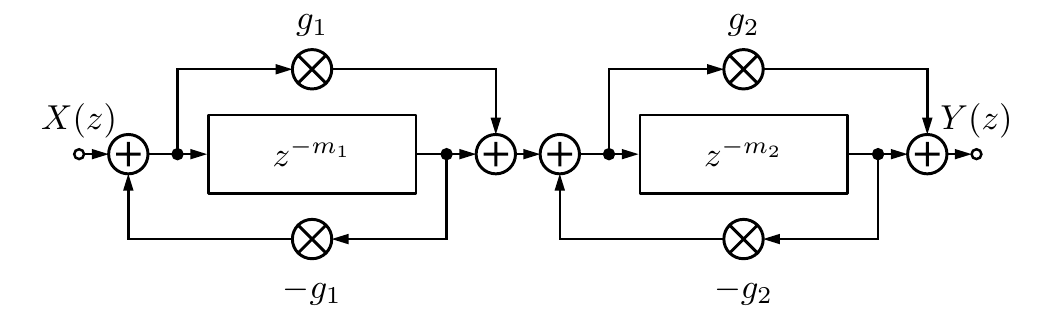}
	\caption{Block diagram of a series of two Schroeder allpasses.}
	\label{fig:allpassCombFilterSeriesBlock}
\end{subfigure}

\vspace{0.4cm}

\begin{subfigure}[b]{0.8\columnwidth}
	\includegraphics[width=\columnwidth]{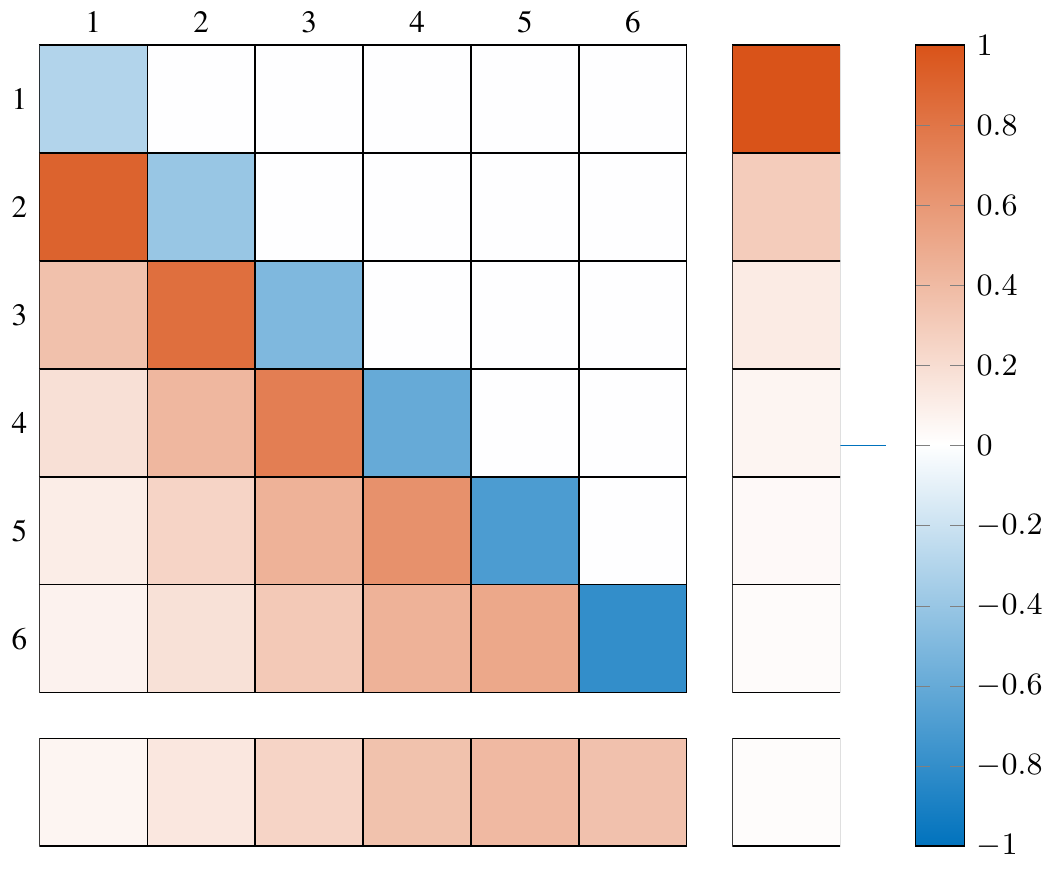}
	\caption{System matrix $\Sys$ in \eqref{eq:systemMatrix} of a series of six Schroeder allpasses with matrix blocks $\Fbm$, $\Ingain$, $\Outgain$, and $\directgain$ as in \eqref{eq:seriesAllpassSS}. The gains $\gain_1, \dots, \gain_6$ are [0.3, 0.4, 0.5, 0.6, 0.7, 0.8].}
	\label{fig:allpassCombFilterSeriesMatrix}
\end{subfigure}
\caption{SISO uniallpass filter based on a series of Schroeder allpasses \cite{Schroeder:1960js}.}
\label{fig:allpassCombFilterSeries}
\end{figure}
%%%%%%%%%%%%%%%%%%%%%%%

\section{Application}
\label{sec:application}

In this section, we show that three well-known delay-based allpass structures are uniallpass FDNs: Schroeder's series allpass \cite{Schroeder:1961ke}, Gardner's nested allpasses \cite{Gardner:1992ks}, and Poletti's unitary reverberator \cite{Poletti:1995tq}. Reviewing these previous designs also reveals their limited design space and demonstrates the significant extension introduced by Theorem~\ref{th:allpassFDN}. We conclude this section by presenting a complete numerical example of a SISO uniallpass FDN with homogeneous decay. \schlecsn{The diagonal similarity matrix $\Dsim$ was computed by solving the discrete Lyapunov equation \eqref{eq:lyapunovB} either numerically or symbolically.}

\subsection{SISO - Series Schroeder Allpass}
%The Schroeder allpass can be given as
%
%\begin{equation}
%	\Tf{Schroeder}(z) = \frac{X(z)}{Y(z)} = \frac{ \gain + z^{-\delay}}{1 + \gain z^{-\delay} }
%	\label{eq:Schroeder}
%\end{equation}
The Schroeder series allpass of $\matSize$ feedforward-feedback delay allpasses is
\begin{equation}
	\Tf{Schroeder}(z) = \prod_{i=1}^\matSize \frac{ \gain_i + z^{-\delay_i}}{1 + \gain_i z^{-\delay_i} },
\end{equation}
where $\gain_i$ and $\delay_i$ denote the feedforward-feedback gains and delay lengths, respectively. Fig.~\ref{fig:allpassCombFilterSeriesBlock} shows an instance for $\matSize = 2$. The corresponding state space realization is \cite{Schlecht:2012uw}
\begin{subequations}
\label{eq:seriesAllpassSS}
\begin{align}
	\fbm_{ij} &= \begin{cases}
		-g_i & \textrm{ for } i = j \\
		0 & \textrm{ for } i < j \\
		\p*{1 - g_j^2} \prod_{k = j+1}^{i-1} g_k & \textrm{ for } i > j
	\end{cases}, \\
	\ingain_i &= \prod_{k = 1}^{i-1} g_k ,  \\
	\outgain_i &= \p*{1 - g_i^2} \prod_{k = i+1}^\matSize g_k ,  \\
	\directgain &= \prod_{k = 1}^\matSize g_k,
\end{align}
\end{subequations}
and the similarity transform $\Dsim$ in \eqref{eq:allpassFDNThcondition} is a diagonal matrix with diagonal elements 
\begin{equation}
	\dsim_{ii} = \frac{1}{1 - g_i^2}.
\end{equation}
Fig.~\ref{fig:allpassCombFilterSeriesMatrix} depicts the system matrix $\Sys$ of the Schroeder series allpass for $\matSize = 6$. The feedback matrix $\Fbm$ is triangular with gains $\gain_1, \dots, \gain_\matSize$ on the main diagonal. The remaining gains $\Ingain$, $\Outgain$, and $\directgain$ are determined by the gains $\gain_i$ as well. Therefore, there exists $\Fbm = \Unitary \Decay$ with triangular unilossless $\Unitary$ and $\Decay = \diag{[\gain_1, \dots, \gain_\matSize]}$ such that the Schroeder series allpass can have homogeneous decay, see \eqref{eq:homoMatrix}.

%As the explicit form is relatively unwieldy, we give MATLAB code below to compute the example.
%\schlecsn{Add code , figure} 

%%%%%%%%%%%%%%%%%%%%%%%
\begin{figure}[!tb]
\centering
\begin{subfigure}[b]{\columnwidth}
	\includegraphics[width=\columnwidth]{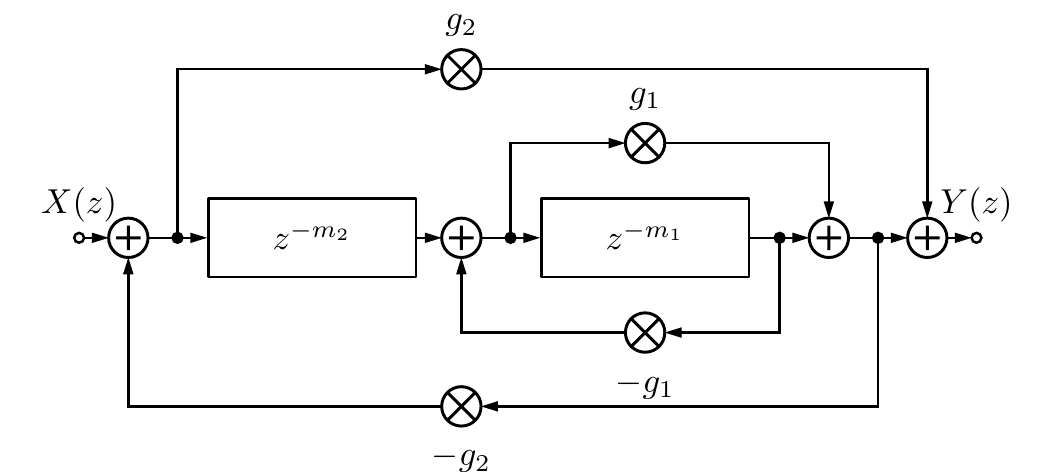}
	\caption{Block diagram of two nested Schroeder allpasses.}
	\label{fig:nestedAllpassBlock}
\end{subfigure}

\vspace{0.4cm}

\begin{subfigure}[b]{0.8\columnwidth}
	\includegraphics[width=\columnwidth]{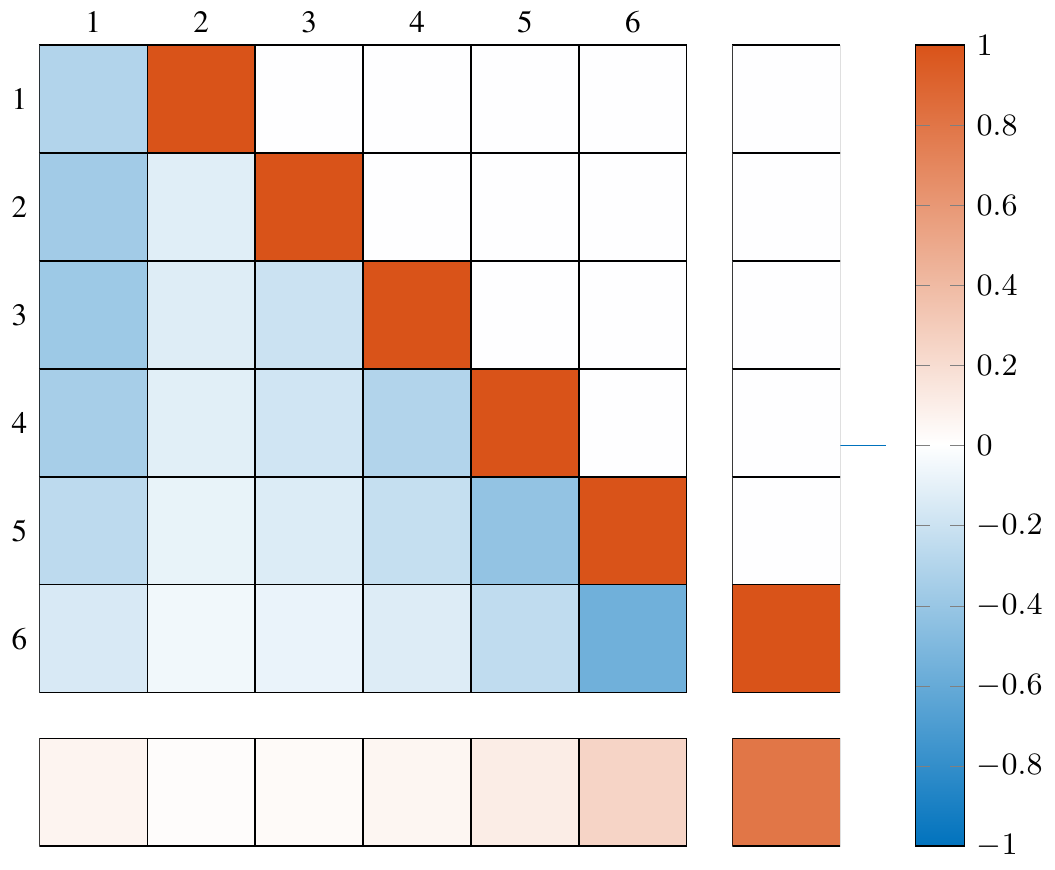}
	\caption{System matrix $\Sys$ in \eqref{eq:systemMatrix} of six nested Schroeder allpasses with matrix blocks $\Fbm$, $\Ingain$, $\Outgain$, and $\directgain$ as in \eqref{eq:nestedAllpassSS}. The gains $\gain_1, \dots, \gain_6$ are [0.3, 0.4, 0.5, 0.6, 0.7, 0.8].}
	\label{fig:nestedAllpassMatrix}
\end{subfigure}
\caption{SISO uniallpass filter based on nested Schroeder allpasses proposed by Gardner \cite{Gardner:1992ks}.}
\label{fig:nestedAllpass}
\end{figure}
%%%%%%%%%%%%%%%%%%%%%%%

\subsection{SISO - Nested Allpass}
The nested allpass as proposed by Gardner \cite{Gardner:1992ks} is a recursive nesting of Schroeder allpasses, i.e., 
\begin{equation}
	\Tf{Gardner} = \Tf{\matSize}(z), 
\end{equation}
where $\Tf{1}(z) = \frac{ \gain_1 + z^{-\delay_1}}{1 + \gain_1 z^{-\delay_1} }$ and for $k > 1$
\begin{equation}
	\Tf{k}(z) = \frac{ \gain_k + z^{-\delay_k} \Tf{k-1}(z)}{1 + \gain_k z^{-\delay_k} \Tf{k-1}(z) }.
\end{equation}
Figure~\ref{fig:nestedAllpassBlock} shows an instance of the nested allpass for $\matSize = 2$. The corresponding state space realization is
\begin{subequations}
\label{eq:nestedAllpassSS}
\begin{align}
	\fbm_{ij} &= \begin{cases}
		- \gain_i \shiftGain_i & \textrm{ for } i = j \\
		1 & \textrm{ for } i = j - 1 \\
		0 & \textrm{ for } i < j - 1 \\
		- \gain_i \shiftGain_j \prod_{k = j}^{i-1} {1 - g_k^2} & \textrm{ for } i > j
	\end{cases}, \\
	\ingain_i &= \begin{cases} 1 \qquad  & \textrm{ for } i = \matSize \\ 0 \qquad &\textrm{ otherwise }\end{cases},  \\
	\outgain_i &= \shiftGain_i \prod_{k = i}^{\matSize} {1 - g_k^2} ,  \\
	\directgain &= \gain_\matSize,
\end{align}
\end{subequations}
where $\shiftGain_1 = 1$ and $\shiftGain_j = \gain_{j-1}$ for $2 \leq j \leq \matSize$. The similarity transform $\Dsim$ in \eqref{eq:allpassFDNThcondition} is a diagonal matrix with diagonal elements 
\begin{equation}
	\dsim_{ii} = \frac{-1}{\prod_{k=i}^\matSize 1 - g_k^2}.
\end{equation}
Fig.~\ref{fig:nestedAllpassMatrix} depicts the system matrix $\Sys$ of the nested allpasses for $\matSize = 6$. The feedback matrix $\Fbm$ is Hessenberg and all gains including $\Ingain$, $\Outgain$, and $\directgain$ are determined by the gains $\gain_i$. Series allpasses are strongly related to nested allpasses as they share the same parameter space, however, differ in the structure. Interestingly, the feedback matrix of nested allpasses induce a much more complex decay pattern than the series allpass counterpart. 

%\schlecsn{add plots}

%%%%%%%%%%%%%%%%%%%%%%%
\begin{figure}[!tb]
\centering
\begin{subfigure}[b]{\columnwidth}
	\includegraphics[width=\columnwidth]{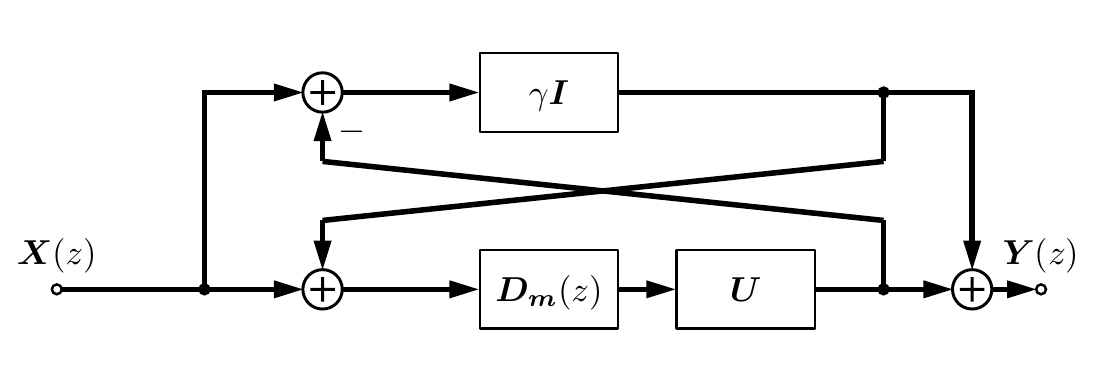}
	\caption{Block diagram of Poletti's unitary reverberator.}
	\label{fig:polettiFDNBlock}
\end{subfigure}

\vspace{0.4cm}

\begin{subfigure}[b]{0.8\columnwidth}
	\includegraphics[width=\columnwidth]{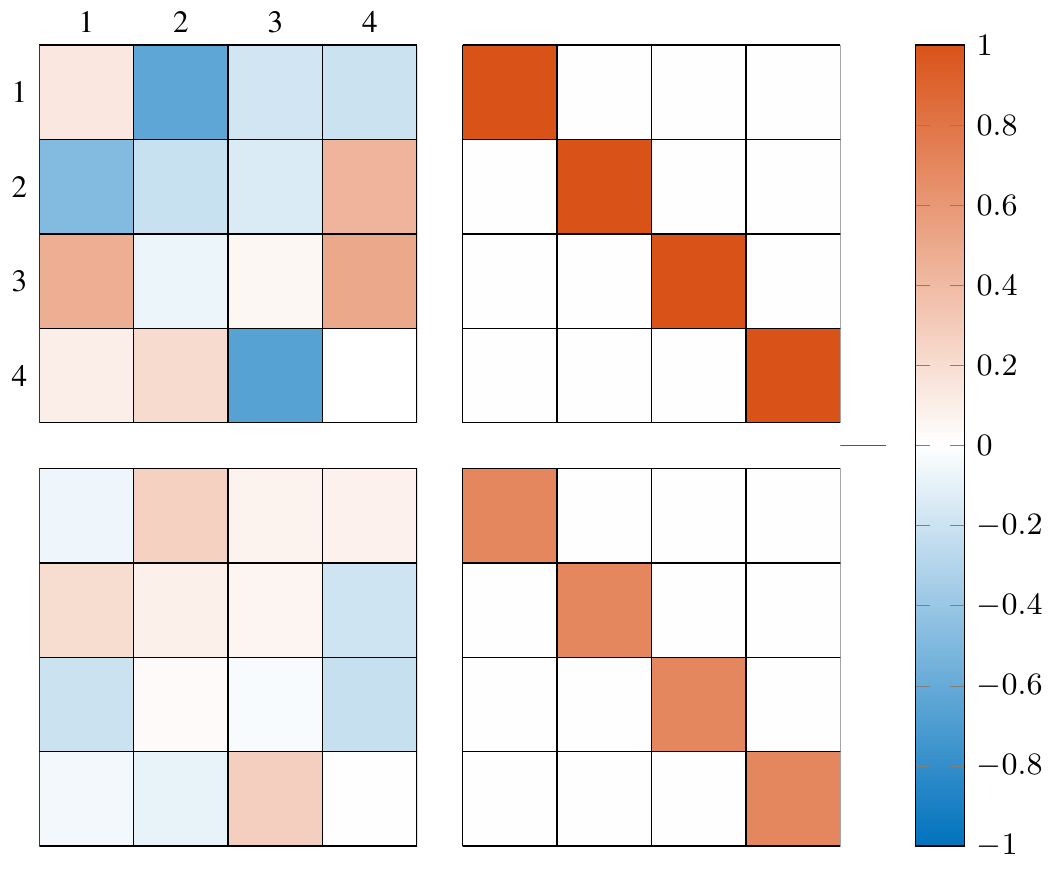}
	\caption{System matrix $\Sys$ in \eqref{eq:systemMatrix} of Poletti's unitary reverberator with matrix blocks $\Fbm$, $\Ingains$, $\Outgains$, and $\Directgains$ as in \eqref{eq:polettiAllpassSS}. The loop gain is $\loopGain = 0.7$ and $\Unitary$ is a random orthogonal $4 \times 4$ matrix.}
	\label{fig:polettiFDNMatrix}
\end{subfigure}
\caption{MIMO uniallpass feedback delay network (FDN) with feedback matrix $\Fbm$ and loop gain $\loopGain$ proposed by Poletti \cite{Poletti:1995tq}. Thick lines indicate multiple channel.}
\label{fig:polettiFDN}
\end{figure}
%%%%%%%%%%%%%%%%%%%%%%%

\subsection{MIMO - Poletti Reverberator}
The MIMO reverberator proposed by Poletti \cite{Poletti:1995tq} is a direct multichannel generalization of the Schroeder allpass structure in lattice form, see Fig.~\ref{fig:polettiFDNBlock}. The loop gain $\loopGain$ controls the decay rate of the response tail such that
\begin{equation}
	\Tf{Poletti}(z) = \p*{ \loopGain \eye + \Unitary \stdDelay } \invp{  \eye + \loopGain \Unitary \stdDelay }.
\end{equation}
The state space realization is
\begin{subequations}
\label{eq:polettiAllpassSS}
\begin{align}
	\Fbm &= - \loopGain \Unitary, \\
	\Ingains &= (1 + \loopGain) \eye, \\
	\Outgains &= (1 - \loopGain) \Unitary, \\
	\Directgains &= \loopGain \eye,
\end{align}
\end{subequations}
and the similarity matrix in \eqref{eq:allpassFDNThcondition} is 
\begin{equation}
	\Dsim = \frac{1+\loopGain}{\sqrt{1 - \loopGain^2}} \eye.
\end{equation}
%The corresponding balanced FDN is 
%orthogonal
%\begin{equation}
%	\Sys = \begin{bmatrix}
%		- \loopGain \Unitary & \sqrt{1 - \loopGain^2} \eye \\
%		\sqrt{1 - \loopGain^2} \Unitary\tran & \loopGain \eye
%	\end{bmatrix}
%\end{equation}
%
%This satisfies Theorem~\ref{th:allpassFDN} and is therefore a uniallpass MIMO allpass FDN.
Fig.~\ref{fig:polettiFDNMatrix} depicts the system matrix $\Sys$ of Poletti's allpass for $\matSize = 4$ and $\numIO = 4$. While the direct and input gains, $\Directgains$ and $\Ingains$, respectively, are scaled identity matrices, the feedback matrix $\Fbm$ and output gains $\Outgains$ are scaled versions of the unitary matrix $\Unitary$. Interestingly, Poletti's allpass has homogeneous decay only for equal delays, \schlecsn{which is usually undesirable as the time-domain response is non-zero only at integer multiples of the delays and can therefore never become dense \cite{Schlecht:2017il}}.

%%%%%%%%%%%%%%%%%%%%%%%
\begin{figure}[!tb]
\centering
\includegraphics[width=0.8\columnwidth]{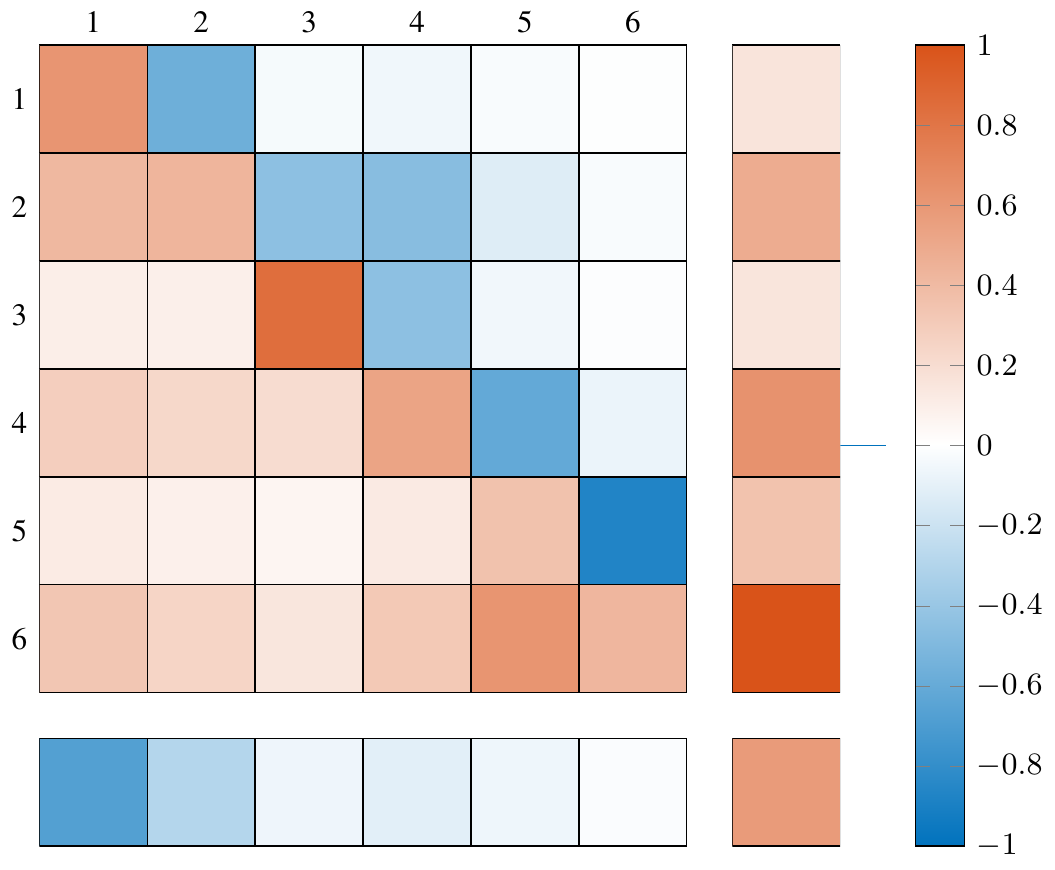}
\caption{System matrix $\Sys$ in \eqref{eq:systemMatrix} of the proposed SISO uniallpass FDN with homogeneous decay with matrix blocks $\Fbm$, $\Ingain$, $\Outgain$, and $\directgain$ as in Section~\ref{sec:homoAllpassSS}. The design parameters are $\matSize = 6$, $\decayRate = 0.99$ and $\Delay = [13,22,1,10,5,3]$.}
\label{fig:HomogeneousSISO}
\end{figure}
%%%%%%%%%%%%%%%%%%%%%%%

\subsection{SISO Homogeneous Decay Uniallpass FDN}
\label{sec:homoAllpassSS}
We give a numerical example of a SISO allpass FDN with homogeneous decay following the procedure in Section~\ref{sec:homogeneous}. Let $\matSize = 6$, $\decayRate = 0.99$ and $\Delay = [13,22,1,10,5,3]$. Then with \eqref{eq:decayGains}, we have
\begin{equation*}
	\Decay = \diag{ \begin{bmatrix} 
0.878 & 0.802 & 0.990 & 0.904 & 0.951 &  0.970 \\ 
\end{bmatrix} }
\end{equation*}
and from \eqref{eq:homoChoice} we can choose
\begin{equation*}
	\Dsim = \diag{ \begin{bmatrix} 
1.000 & 1.808 & 2.096 & 2.743 & 3.413 &  3.662 \\ 
\end{bmatrix} }.
\end{equation*}
From \eqref{eq:unitaryAlpha}, we can then compute
\begin{equation*}
	\Unitary = \begin{bmatrix}
0.702 & -0.708 & -0.034 & -0.059 & -0.027 &  -0.006 \\ 
0.474 & 0.540 & -0.448 & -0.515 & -0.132 &  -0.026 \\ 
0.120 & 0.120 & 0.853 & -0.491 & -0.055 &  -0.010 \\ 
0.327 & 0.289 & 0.210 & 0.589 & -0.642 &  -0.078 \\ 
0.136 & 0.114 & 0.059 & 0.141 & 0.378 &  -0.896 \\ 
0.378 & 0.310 & 0.152 & 0.352 & 0.651 &  0.437 \\ 
	\end{bmatrix} .
\end{equation*}
The feedback matrix results then from \eqref{eq:homoMatrix}, i.e., 
\begin{equation*}
	\Fbm = \begin{bmatrix} 
0.616 & -0.568 & -0.034 & -0.054 & -0.025 &  -0.005 \\ 
0.416 & 0.433 & -0.443 & -0.466 & -0.125 &  -0.025 \\ 
0.105 & 0.097 & 0.844 & -0.444 & -0.052 &  -0.010 \\ 
0.287 & 0.232 & 0.208 & 0.533 & -0.611 &  -0.076 \\ 
0.120 & 0.091 & 0.059 & 0.127 & 0.360 &  -0.869 \\ 
0.332 & 0.249 & 0.151 & 0.318 & 0.619 &  0.424 \\ 
\end{bmatrix}.
\end{equation*}
The remaining input, output and direct gains are determined by solving the completion problem in Section~\ref{sec:generalCompletion}
\begin{align*}
	\Ingain\tran &= \begin{bmatrix} 0.159 & 0.483 & 0.156 & 0.633 & 0.354 &  1.073 \\   \end{bmatrix},  \\
	\Outgain &= -\begin{bmatrix} 0.675 & 0.290 & 0.064 & 0.109 & 0.062 &  0.014 \\   \end{bmatrix},  \\
	\directgain &= 0.581.
\end{align*}
Fig.~\ref{fig:HomogeneousSISO} shows the system matrix for the numerical example. Interestingly, the feedback matrix $\Fbm$ exhibits a triangular-like shape which suggests that the homogeneous decay uniallpass FDN generalizes the triangular and Hessenberg shapes of the series and nested allpasses. 

%A modal decomposition of this FDN confirms that all system $\pole_i$ have the same magnitude, i.e., $\abs{\pole_i} = \decayRate$ \cite{Schlecht:2019uj}. \todo{add?}

%\subsection{Reverberation Enhancement}
%Time-varying MIMO Allpass FDN

\section{Conclusion}
In this work, we developed a novel characterization for uniallpass feedback delay networks (FDNs), which are allpass for any choice of delay lengths. Further, we introduced the uniallpass completion, i.e., completing a given feedback matrix to a uniallpass FDN. While the full MIMO case is relatively simple, also a solution to the SISO case was presented. Further, we solved the completion problem for a particular class of feedback matrices, which yields homogeneous decay of the impulse response. We reviewed three previous allpass FDN designs within this novel characterization and an additional numerical example for homogeneous decay uniallpass FDNs.

Future research questions should address application-specific designs of uniallpass FDNs, for instance, in audio signal processing, where additional constraints are required. Further research is also needed for the design of frequency-dependent FDN designs with the allpass property\schlecsn{, i.e., for a filter feedback matrix $\Fbm(z)$. In particular, the homogeneous decay allpass FDN with filter matrix $\Decay(z)$ in \eqref{eq:homoMatrix} has important practical applications for frequency-dependent decay and generalizes the single delay case, i.e., $\matSize = 1$ in \cite{Schlecht:2019sa}.}

\section{Acknowledgement}
The author thanks Prof. Dario Fasino for his insights on orthogonal Cauchy-like matrices in Section~\ref{sec:homogeneousSISO}. Further thanks go to Dr. Maximilian Schäfer and Prof. Vesa Välimäki for proofreading and valuable comments. The author is grateful to the anonymous reviewers for their detailed and thorough comments, which helped improve this manuscript.

% Either list references using the bibliography style file IEEEtran.bst
\bibliographystyle{IEEEtranNoURL}
\bibliography{Papers}

\begin{thebibliography}{10}
\providecommand{\url}[1]{#1}
\csname url@samestyle\endcsname
\providecommand{\newblock}{\relax}
\providecommand{\bibinfo}[2]{#2}
\providecommand{\BIBentrySTDinterwordspacing}{\spaceskip=0pt\relax}
\providecommand{\BIBentryALTinterwordstretchfactor}{4}
\providecommand{\BIBentryALTinterwordspacing}{\spaceskip=\fontdimen2\font plus
\BIBentryALTinterwordstretchfactor\fontdimen3\font minus
  \fontdimen4\font\relax}
\providecommand{\BIBforeignlanguage}[2]{{%
\expandafter\ifx\csname l@#1\endcsname\relax
\typeout{** WARNING: IEEEtran.bst: No hyphenation pattern has been}%
\typeout{** loaded for the language `#1'. Using the pattern for}%
\typeout{** the default language instead.}%
\else
\language=\csname l@#1\endcsname
\fi
#2}}
\providecommand{\BIBdecl}{\relax}
\BIBdecl

\bibitem{Regalia:1988cj}
P.~Regalia, S.~Mitra, and P.~Vaidyanathan, ``{The digital all-pass filter: a
  versatile signal processing building block},'' \emph{Proceedings of the
  IEEE}, vol.~76, no.~1, pp. 19 -- 37, 1988.

\bibitem{Schroeder:1961ke}
M.~R. Schroeder and B.~F. Logan, ``{"Colorless" artificial reverberation},''
  \emph{IRE Transactions on Audio}, vol. AU-9, no.~6, pp. 209 -- 214, 1961.

\bibitem{Gerzon:1971tu}
M.~A. Gerzon, ``{Synthetic stereo reverberation: Part One},'' vol.~13, pp. 632
  -- 635, 1971.

\bibitem{Gerzon:1976fm}
------, ``\BIBforeignlanguage{English}{{Unitary (energy-preserving)
  multichannel networks with feedback}},''
  \emph{\BIBforeignlanguage{English}{Electronics Letters}}, vol.~12, no.~11,
  pp. 278 -- 279, 1976.

\bibitem{Rocchesso:1997fv}
D.~Rocchesso and J.~Smith, ``{Circulant and elliptic feedback delay networks
  for artificial reverberation},'' \emph{IEEE Transactions on Speech and Audio
  Processing}, vol.~5, no.~1, pp. 51 -- 63, 1997.

\bibitem{Schlecht:2017jt}
S.~J. Schlecht and E.~A.~P. Habets, ``{On Lossless Feedback Delay Networks},''
  \emph{IEEE Transactions on Signal Processing}, vol.~65, no.~6, pp. 1554 --
  1564, 2016.

\bibitem{Jot:1991tq}
J.~M. Jot and A.~Chaigne, ``{Digital delay networks for designing artificial
  reverberators},'' ser. Proc. Audio Eng. Soc. Conv., Paris, France, 1991, pp.
  1 -- 12.

\bibitem{Prawda:2019tq}
K.~Prawda, S.~J. Schlecht, and V.~Välimäki, ``{Improved Reverberation Time
  Control for Feedback Delay Networks},'' ser. Proc. Int. Conf. Digital Audio
  Effects (DAFx), 2019, pp. 1 -- 7.

\bibitem{DeSena:2015bb}
\BIBentryALTinterwordspacing
E.~D. Sena, H.~Hacıhabiboğlu, Z.~Cvetkovic, and J.~O.~S. III, ``{Efficient
  synthesis of room acoustics via scattering delay networks},'' \emph{IEEE/ACM
  Trans. Audio, Speech, Language Process.}, vol.~23, no.~9, pp. 1478 -- 1492,
  2015.
\BIBentrySTDinterwordspacing

\bibitem{Schlecht:2017il}
\BIBentryALTinterwordspacing
S.~J. Schlecht and E.~A.~P. Habets, ``{Feedback delay networks: Echo density
  and mixing time},'' \emph{IEEE/ACM Trans. Audio, Speech, Language Process.},
  vol.~25, no.~2, pp. 374 -- 383, 2017.
\BIBentrySTDinterwordspacing

\bibitem{Schlecht:2019uj}
------, ``{Modal Decomposition of Feedback Delay Networks},'' \emph{IEEE
  Transactions on Signal Processing}, vol.~67, no.~20, pp. 5340--5351, 2019.

\bibitem{Schlecht:2019sa}
S.~J. Schlecht, ``{Frequency-Dependent Schroeder Allpass Filters},''
  \emph{Applied Sciences}, vol.~10, no.~1, p. 187, 2019.

\bibitem{Gardner:1992ks}
W.~G. Gardner, ``\BIBforeignlanguage{English}{{A real-time multichannel room
  simulator}},'' \emph{\BIBforeignlanguage{English}{J. Acoust. Soc. Am.}},
  vol.~92, no.~4, pp. 1 -- 23, 1992.

\bibitem{Vaidyanathan:1989}
P.~P. Vaidyanathan and Z.~Doganata, ``{The role of lossless systems in modern
  digital signal processing: a tutorial},'' \emph{IEEE Transactions on
  Education}, vol.~32, no.~3, pp. 181--197, 1989.

\bibitem{Poletti:1995tq}
\BIBentryALTinterwordspacing
M.~A. Poletti, ``\BIBforeignlanguage{English}{{A Unitary Reverberator For
  Reduced Colouration In Assisted Reverberation Systems}},'' ser. INTER-NOISE
  and NOISE-CON, vol.~5, Newport Beach, CA, USA, 1995, pp. 1223 -- 1232.
\BIBentrySTDinterwordspacing

\bibitem{Valimaki:2012jv}
V.~Välimäki, J.~D. Parker, L.~Savioja, J.~O.~S. III, and J.~S. Abel, ``{Fifty
  years of artificial reverberation},'' \emph{IEEE/ACM Trans. Audio, Speech,
  Language Process.}, vol.~20, no.~5, pp. 1421 -- 1448, 2012, readingList.

\bibitem{Vaananen:1997wj}
R.~Väänänen, V.~Välimäki, J.~Huopaniemi, and M.~Karjalainen, ``{Efficient
  and Parametric Reverberator for Room Acoustics Modeling},'' ser. Proc. Int.
  Comput. Music Conf., Thessaloniki, Greece, 1997, pp. 200 -- 203.

\bibitem{Lokki:2001ly}
T.~Lokki and J.~Hiipakka, ``{A time-variant reverberation algorithm for
  reverberation enhancement systems},'' ser. Proc. Int. Conf. Digital Audio
  Effects (DAFx), Limerick, Ireland, 2001, pp. 28 -- 32.

\bibitem{Schlecht:2015hi}
S.~J. Schlecht and E.~A.~P. Habets,
  ``\BIBforeignlanguage{English}{{Time-varying feedback matrices in feedback
  delay networks and their application in artificial reverberation}},''
  \emph{\BIBforeignlanguage{English}{J. Acoust. Soc. Am.}}, vol. 138, no.~3,
  pp. 1389 -- 1398, 2015.

\bibitem{Werner.2020.Energy}
K.~J. Werner, ``{Energy-Preserving Time-Varying Schroeder Allpass Filters},''
  in \emph{Proceedings of the 23rd International Conference on Digital Audio
  Effects (DAFx2020)}, Vienna, Austria, 2020.

\bibitem{Kendall:1995be}
\BIBentryALTinterwordspacing
G.~S. Kendall, ``{The Decorrelation of Audio Signals and Its Impact on Spatial
  Imagery},'' \emph{Comput. Music J.}, vol.~19, no.~4, pp. 71--87, 1995.
\BIBentrySTDinterwordspacing

\bibitem{Abel:2019aa}
J.~S. Abel and E.~K. Canfield-Dafilou, ``{Dispersive Delay and Comb Filters
  Using a Modal Structure},'' \emph{IEEE Signal Processing Letters}, vol.~26,
  no.~12, pp. 1748--1752, 2019.

\bibitem{Gribben:2020}
C.~Gribben and H.~Lee, ``{The Perception of Band-Limited Decorrelation Between
  Vertically Oriented Loudspeakers},'' \emph{IEEE/ACM Transactions on Audio,
  Speech, and Language Processing}, vol.~28, pp. 876--888, 2020.

\bibitem{Poletti:2004hh}
M.~A. Poletti, ``\BIBforeignlanguage{English}{{The Stability Of Multichannel
  Sound Systems With Frequency Shifting}},''
  \emph{\BIBforeignlanguage{English}{J. Acoust. Soc. Am.}}, vol. 116, no.~2,
  pp. 853 -- 871, 2004.

\bibitem{Schlecht:2016ta}
S.~J. Schlecht and E.~A.~P. Habets, ``\BIBforeignlanguage{English}{{The
  stability of multichannel sound systems with time-varying mixing
  matrices}},'' \emph{\BIBforeignlanguage{English}{J. Acoust. Soc. Am.}}, vol.
  140, no.~1, pp. 601 -- 609, 2016.

\bibitem{Parker:2013fy}
J.~Parker and V.~Välimäki, ``{Linear Dynamic Range Reduction of Musical Audio
  Using an Allpass Filter Chain},'' \emph{IEEE Signal Processing Letters},
  vol.~20, no.~7, pp. 669 -- 672, 2013.

\bibitem{Belloch:2014}
J.~A. Belloch, J.~Parker, L.~Savioja, A.~Gonzalez, and V.~Välimäki,
  ``{Dynamic range reduction of audio signals using multiple allpass filters on
  a GPU accelerator},'' in \emph{2014 22nd European Signal Processing
  Conference (EUSIPCO)}, 2014, pp. 890--894.

\bibitem{Abel:2006ux}
J.~S. Abel and J.~O.~S. III, ``{Robust Design of Very High-Order Allpass
  Dispersion Filters},'' ser. Proc. Int. Conf. Digital Audio Effects (DAFx),
  Montreal, QC, Canada, 2006, pp. 13 -- 18.

\bibitem{Valimaki:2010wc}
V.~Välimäki, J.~D. Parker, and J.~S. Abel,
  ``\BIBforeignlanguage{English}{{Parametric Spring Reverberation Effect}},''
  \emph{\BIBforeignlanguage{English}{J. Audio Eng. Soc.}}, vol.~58, no. 7/8,
  pp. 547 -- 562, 2010.

\bibitem{Parker:2011fn}
J.~Parker, ``\BIBforeignlanguage{English}{{Efficient Dispersion Generation
  Structures for Spring Reverb Emulation}},''
  \emph{\BIBforeignlanguage{English}{EURASIP Journal on Advances in Signal
  Processing}}, vol. 2011, no.~1, pp. 547 -- 8, 2011.

\bibitem{Kaszkurewicz:2000}
E.~Kaszkurewicz and A.~Bhaya, \emph{{Matrix Diagonal Stability in Systems and
  Computation}}.\hskip 1em plus 0.5em minus 0.4em\relax Birkhäuser Basel,
  2000.

\bibitem{Baggio:2015.Factorization}
G.~Baggio and A.~Ferrante, ``{On the Factorization of Rational Discrete-Time
  Spectral Densities},'' \emph{IEEE Transactions on Automatic Control},
  vol.~61, no.~4, pp. 969--981, 2015.

\bibitem{Baggio:2017.SpectralDensity}
------, ``{Parametrization of Minimal Spectral Factors of Discrete-Time
  Rational Spectral Densities},'' \emph{IEEE Transactions on Automatic
  Control}, vol.~64, no.~1, pp. 396--403, 2017.

\bibitem{Regalia:1987fo}
P.~Regalia, S.~Mitra, and J.~Fadavi-Ardekani,
  ``\BIBforeignlanguage{English}{{Implementation of real coefficient digital
  filters using complex arithmetic}},''
  \emph{\BIBforeignlanguage{English}{Circuits and Systems, IEEE Transactions
  on}}, vol.~34, no.~4, pp. 345 -- 353, 1987.

\bibitem{Mullis:1976.Roundoff}
C.~Mullis and R.~Roberts, ``{Synthesis of minimum roundoff noise fixed point
  digital filters},'' \emph{IEEE Transactions on Circuits and Systems},
  vol.~23, no.~9, pp. 551--562, 1976.

\bibitem{Hanzon:2000}
B.~Hanzon and R.~L.~M. Peeters, ``{Balanced Parametrizations of Stable SISO
  All-Pass Systems in Discrete Time},'' \emph{Mathematics of Control, Signals
  and Systems}, vol.~13, no.~3, pp. 240--276, 2000.

\bibitem{Schlecht.2020.FDNTB}
S.~J. Schlecht, ``{FDNTB: The Feedback Delay Network Toolbox},'' in
  \emph{Proceedings of the 23rdInternational Conference on Digital Audio
  Effects (DAFx2020)}, Vienna, Austria, 2020, pp. 211--218.

\bibitem{SmithIII:2007tm}
\BIBentryALTinterwordspacing
J.~O.~S. III, \emph{{Introduction to Digital Filters with Audio Applications}},
  ser. W3K Publishing.\hskip 1em plus 0.5em minus 0.4em\relax W3K Publishing,
  2007.
\BIBentrySTDinterwordspacing

\bibitem{Vaidyanathan:1993uh}
P.~P. Vaidyanathan, \emph{{Multirate Systems and Filter Banks}}, ser. Prentice
  Hall.\hskip 1em plus 0.5em minus 0.4em\relax Prentice Hall, 1993.

\bibitem{Ferrante:2016rf}
A.~Ferrante and G.~Picci, ``{Representation and Factorization of Discrete-Time
  Rational All-Pass Functions},'' \emph{IEEE Transactions on Automatic
  Control}, vol.~62, no.~7, pp. 3262--3276, 2016.

\bibitem{Brualdi:1983ih}
R.~A. Brualdi and H.~Schneider, ``\BIBforeignlanguage{English}{{Determinantal
  identities: Gauss, Schur, Cauchy, Sylvester, Kronecker, Jacobi, Binet,
  Laplace, Muir, and Cayley}},'' \emph{\BIBforeignlanguage{English}{Linear
  Algebra Appl.}}, vol. 52-53, pp. 769 -- 791, 1983.

\bibitem{Loewy:1986dg}
R.~Loewy, ``\BIBforeignlanguage{English}{{Principal minors and diagonal
  similarity of matrices}},'' \emph{\BIBforeignlanguage{English}{Linear Algebra
  and its Applications}}, vol.~78, pp. 23 -- 64, 1986.

\bibitem{Petersen:2012up}
\BIBentryALTinterwordspacing
K.~B. Petersen and M.~S. Pedersen, \emph{\BIBforeignlanguage{English}{{The
  matrix cookbook}}}.\hskip 1em plus 0.5em minus 0.4em\relax Tech. Univ.
  Denmark, Kongens Lyngby, 2012.
\BIBentrySTDinterwordspacing

\bibitem{Engel:1982}
G.~M. Engel and H.~Schneider, ``\BIBforeignlanguage{English}{{Algorithms for
  Testing the Diagonal Similarity of Matrices and Related Problems}},''
  \emph{\BIBforeignlanguage{English}{SIAM. J. on Algebraic and Discrete
  Methods}}, vol.~3, no.~4, pp. 429 -- 438, 1982.

\bibitem{Fiedler:2009}
M.~Fiedler, ``{Suborthogonality and orthocentricity of matrices},''
  \emph{Linear Algebra and its Applications}, vol. 430, no.~1, pp. 296--307,
  2009.

\bibitem{Chu:2000ez}
\BIBentryALTinterwordspacing
M.~T. Chu, ``\BIBforeignlanguage{English}{{A Fast Recursive Algorithm for
  Constructing Matrices with Prescribed Eigenvalues and Singular Values}},''
  \emph{\BIBforeignlanguage{English}{SIAM Journal on Numerical Analysis}},
  vol.~37, no.~3, pp. 1004--1020, 2000.
\BIBentrySTDinterwordspacing

\bibitem{Li:2001be}
\BIBentryALTinterwordspacing
C.-K. Li and R.~Mathias, ``\BIBforeignlanguage{English}{{Construction of
  Matrices with Prescribed Singular Values and Eigenvalues}},''
  \emph{\BIBforeignlanguage{English}{BIT Numerical Mathematics}}, vol.~41,
  no.~1, pp. 115 -- 126, 2001.
\BIBentrySTDinterwordspacing

\bibitem{Fasino:2002}
D.~Fasino and L.~Gemignani, ``{A Lanczos-type algorithm for the QR
  factorization of regular Cauchy matrices},'' \emph{Numerical Linear Algebra
  with Applications}, vol.~9, no.~4, pp. 305--319, 2002.

\bibitem{Schechter:1959}
S.~Schechter, ``{On the Inversion of Certain Matrices},'' \emph{Mathematical
  Tables and Other Aids to Computation}, vol.~13, no.~66, p.~73, 1959.

\bibitem{Rahman:2002tq}
Q.~I. Rahman and G.~Schmeisser, \emph{\BIBforeignlanguage{English}{{Analytic
  Theory of Polynomials}}}.\hskip 1em plus 0.5em minus 0.4em\relax Oxford
  University Press, 2002.

\bibitem{Schroeder:1960js}
M.~R. Schroeder, ``\BIBforeignlanguage{English}{{“Colorless” Artificial
  Reverberation}},'' \emph{\BIBforeignlanguage{English}{J. Acoust. Soc. Am.}},
  vol.~32, no.~11, p. 1520, 1960.

\bibitem{Schlecht:2012uw}
\BIBentryALTinterwordspacing
S.~J. Schlecht and E.~A.~P. Habets, ``{Connections between parallel and serial
  combinations of comb filters and feedback delay networks},'' ser.
  International Workshop on Acoustic Signal Enhancement (IWAENC), 2012, pp. 1
  -- 4.
\BIBentrySTDinterwordspacing

\end{thebibliography}

\end{sloppy}
\end{document}